\documentclass[11pt]{article}

\usepackage{jheppub}
\usepackage{listings}
\usepackage{amsfonts}
\usepackage{amsmath}
\usepackage{amssymb}
\usepackage{amsthm}
\usepackage[mathscr]{eucal}
\usepackage{bbold}
\usepackage{braket}
\usepackage{color}
\usepackage{dsfont}
\usepackage{framed}
\usepackage{mathtools}
\usepackage{physics}
\usepackage[normalem]{ulem}
\usepackage{tensor}
\usepackage{thmtools}
\usepackage{thm-restate}
\usepackage{ulem}
\usepackage{tikz}
\usetikzlibrary{arrows.meta,arrows} 
\usepackage{centernot} 
\usepackage[shortlabels]{enumitem}
\usepackage[a]{esvect}  
\usepackage{slashed}
\usepackage{booktabs}
\usepackage{colortbl}

\usepackage[only,llbracket,rrbracket]{stmaryrd} 

\usepackage{yfonts}
\usepackage{float}
\usetikzlibrary{math} 
\usetikzlibrary{calc} 

\usepackage{cleveref}

\newtheorem{cor}{Corollary}
\newtheorem{lemma}{Lemma}
\newtheorem{thm}{Theorem}

\theoremstyle{definition}
\newtheorem{defi}{Definition}

\newtheorem{eg}{Example}
\newtheorem{remark}{Remark}

\newcommand{\R}{\mathbb{R}}

\newcommand{\N}{{\sf N}}
\newcommand{\D}{\sf{D}}

\newcommand{\mb}{\mathbf}
\newcommand{\leftm}{\mb{L}}
\newcommand{\rightm}{\mb{R}}
\newcommand{\leftocc}{\mb{L}}
\newcommand{\rightocc}{\mb{R}}
\newcommand{\ham}[2]{|#1 - #2|}
\newcommand{\pairset}{\mathcal{P}_\N}
\newcommand{\sbset}{\mathcal{P}_\N^{\rm sup}}
\newcommand{\cset}{\mathcal{P}_\N^{\rm cen}}
\newcommand{\bset}{\mathcal{P}_\N^{\rm bal}}
\newcommand{\mL}{m_{_{\text L}}}
\newcommand{\mR}{m_{_{\text R}}}
\newcommand{\CPI}{cpHEI}

\usepackage{calligra}
\DeclareMathAlphabet{\mathcalligra}{T1}{calligra}{m}{n}
\DeclareFontShape{T1}{calligra}{m}{n}{<->s*[2.2]callig15}{}

\newcommand{\ent}[1]{\mathsf{S}(#1)} 

\newcommand{\be}{\begin{equation}}
\newcommand{\ee}{\end{equation}}

\newtheorem{innercustomconj}{Conjecture}
\newenvironment{customconj}[1]{\renewcommand\theinnercustomconj{#1}\innercustomconj}
  {\endinnercustomconj}

\tikzstyle{startstop} = [rectangle, rounded corners, minimum width=3cm, minimum height=1cm,text centered, draw=black, fill=red!30]

\tikzstyle{process} = [rectangle, minimum width=3cm, minimum height=1cm, text centered, draw=black, fill=orange!30]
\tikzstyle{decision} = [diamond, minimum width=3cm, minimum height=1cm, text centered, draw=black, fill=green!30]

\definecolor{brg}{RGB}{70, 255,70}
\definecolor{darkgreen}{rgb}{0.1, 0.6, 0.2}

\preprint{BRX-TH-6729}

\title{\boldmath Combinatorial properties of holographic entropy inequalities}

\author[a]{Guglielmo Grimaldi,}
\author[a,b]{Matthew Headrick,} 
\author[c]{Veronika E.\ Hubeny,}
\author[d]{and Pavel Shteyner}

\affiliation[a]{Martin Fisher School of Physics, Brandeis University, Waltham MA 02453, USA}
\affiliation[b]{Institut des Hautes Etudes Scientifiques, 91440 Bures-sur-Yvette, France}
\affiliation[c]{Center for Quantum Mathematics and Physics (QMAP)\\ Department of Physics \& Astronomy, University of California, Davis CA 95616, USA}
\affiliation[d]{Department of Mathematics, Technion - Israel Institute of Technology, Haifa, 32000, Israel}
\emailAdd{ggrimaldi@brandeis.edu}
\emailAdd{headrick@brandeis.edu}
\emailAdd{veronika@physics.ucdavis.edu}
\emailAdd{p.shteyner@campus.technion.ac.il}

\abstract{A holographic entropy inequality (HEI) is a linear inequality obeyed by Ryu-Takayanagi holographic entanglement entropies, or equivalently by the minimum cut function on weighted graphs. We establish a new combinatorial framework for studying contraction-provable HEIs (cpHEIs), and use it to prove several properties they share, including two majorization-related properties as well as a necessary and sufficient condition for an inequality to be a cpHEI. We thereby address all the conjectures presented in \cite{Grimaldi:2025jad}, disproving two of them by explicit counterexamples and proving the other two under the contraction-provability assumption. In particular, we show that the null reduction of any superbalanced cpHEI passes the majorization test defined in \cite{Grimaldi:2025jad}, thereby providing strong new evidence that all HEIs are obeyed in time-dependent holographic states.
\newline
\newline
A video abstract is available at \url{https://youtu.be/vUtqPSxG1Io}.
}

\begin{document}

\maketitle

\section{Introduction}

This paper presents a set of new results in combinatorial majorization theory that are motivated by a problem in graph theory, which in turn is motivated by questions in quantum gravity. We have endeavored to write the paper so that it is self-contained, i.e.\ readable without specialized knowledge in any of those fields. In subsec.\ \ref{ssec:broad-overview}, we provide a review of the big picture that attempts to bridge the terminology in separate communities. In subsec.\ \ref{ssec:summary}, we summarize the results of the paper \cite{Grimaldi:2025jad}, which motivated the current work, and in subsec.\ \ref{ssec:results} we summarize the contents of this paper.

\subsection{Broad overview}\label{ssec:broad-overview}

The notion of entropy is a cornerstone of classical information theory, quantum mechanics, and, more recently, quantum gravity. In each of these areas, the entropy function shares many basic features. For example, dividing a system into $\N$ subsystems labelled $1,\ldots,\N$, the entropy is always a submodular set-function $F: 2^{\Omega} \to \mathbb{R}_+$, where $\mathbb{R}_+$ is the non-negative reals and $\Omega =[\N]:=\{1,2,3,\dots,\N\}$; i.e.\ for all $X,Y \in 2^\Omega$ we have
\begin{equation}\label{eq:submod}
    F(X) + F(Y) \geq F(X \cap Y) + F(X \cup Y)\,.
\end{equation}
This submodularity condition is also known in the physics literature as \emph{strong subadditivity} (SSA). Yet, entropy functions are less general than submodular functions: they obey additional properties. Tracing how these properties change from one context to another has been key in understanding the nature of information itself, from classical correlations to quantum entanglement and, surprisingly, to the emergence of spacetime.

\paragraph{Classical and quantum entropies:} In classical information theory, $\Omega$ indexes a collection of jointly distributed random variables $t_1\in T_1,\ldots,t_\N\in T_\N$. For $X \in 2^{\Omega}$, with the corresponding set $T_X:=\prod_{i\in X}T_i$ and marginal probability distribution $p_X$, the \emph{Shannon entropy} is defined as
\begin{equation}
    H(X) := - \sum_{t\in T_X}p_X(t)\log(p_X(t))\,.
\end{equation}
The Shannon entropy is monotone non-decreasing, i.e.\ $H(X\cup Y) \geq H(X)$ for $X,Y \in 2^\Omega$. Furthermore, it obeys additional non-trivial constraints (other than submodularity); see \cite{zhang:1998, 1023507,4557201}. 

Crucially, the monotonicity property fails in the quantum setting where $\Omega$ indexes a collection of reduced density operators. More specifically, given a density operator $\rho$ acting on a Hilbert space $\mathcal{H} =\mathcal{H}_1 \otimes \mathcal{H}_2 \otimes \cdots \otimes \mathcal{H}_{\N}$, the \emph{von Neumann entropy} $\mathsf{S}$ is defined analogously to its classical counterpart:
\begin{equation}
    \ent{X} := \ent{\rho_X}:= -\Tr \rho_X \log \rho_X\,,
\end{equation}
where $\rho_X := \Tr_{\mathcal{H}_{[\N]\setminus X}}\rho$ and 
$\mathcal{H}_{[\N]\setminus X} = \bigotimes_{i\notin X} \mathcal{H}_i$.
As mentioned, the von Neumann entropy is not monotone non-decreasing; instead, it can be shown to be \emph{weakly monotonic}:
\begin{equation}\label{eq:wmo}
    \ent{X} + \ent{Y} \geq \ent{X \setminus Y} + \ent{Y \setminus X}\,.
\end{equation}
The fact that entropies of density operators can violate monotonicity is a manifestation of entanglement, a type of quantum correlation that is inexplicable by classical physics. An open question in quantum information theory is whether the von Neumann entropy obeys further constraints, particularly in the form of linear inequalities. For $\mathsf{N} = 3$, it is known that \eqref{eq:submod} and \eqref{eq:wmo} are the only linear constraints on $\mathsf{S}$, but the question remains open for $\mathsf{N} \geq 4$.\footnote{For $\mathsf{N} \geq 4$, only \emph{constrained} inequalities are known \cite{Linden:2004ebt, Cadney:2011vix, Christandl:2023zry}. These are inequalities that hold only for a class of density operators that saturate other constraints.}

\paragraph{Entropies in quantum gravity:} The story becomes even richer when jumping into the realm of quantum gravity. This is the setting of this paper; however, we stress that no knowledge of quantum gravity is required for the results shown here. The essential part is that we are studying the von Neumann entropy for a special class of density operators that are of interest in quantum gravity; specifically, static states in the AdS/CFT correspondence. Stripping away many physical details that are unimportant for our purposes, we can summarize the situation as follows. The state admits a representation in terms of a Riemannian manifold with boundary $\Sigma$, with the subsystems $1,\ldots,\N$ represented by non-overlapping regions $A_1,\dots,A_\N$ of the boundary $\partial\Sigma$. For any set $X\subset[\N]$ of regions, the entropy $\ent{X}$ is computed by the area of the minimal hypersurface in $\Sigma$ homologous to $A_X:=\cup_{i\in X}A_i$:
\be\label{eq:RT}
\ent{X} = \frac1{4G_{\rm N}}\min_{\gamma\sim A_X}|\gamma|\,,
\ee
where $G_{\rm N}$ is a positive constant (the Newton constant, or equivalently Planck area), $\gamma\sim A_X$ means that the hypersurface $\gamma$ in $\Sigma$ is homologous to $A_X$ (in other words, there exists a region of $\Sigma$ bounded by $A_X\cup\gamma$), and $|\cdot|$ denotes the area. Eq.\ \eqref{eq:RT} is the Ryu-Takayanagi (RT) formula \cite{Ryu:06b20v,Ryu:2006ef}.

It was shown in \cite{Bao:2015bfa} that the application of the RT formula for a given $\Sigma$ can be reduced to the problem of computing minimal cuts on a certain weighted undirected graph $G$ with $\N+1$ terminals (or ``external vertices''). Each boundary region $A_i$ corresponds to the terminal labelled $i$, and the complementary region $\partial\Sigma\setminus\cup_{i\in[\N]}A_i$ corresponds to the terminal labelled $\N+1$. Then $\ent{X}$ equals the min cut on $G$ between the terminals with labels in $X$ and the rest. Conversely, for every weighted undirected graph with $\N+1$ terminals, there exists a Riemannian manifold with boundary $\Sigma$ and boundary regions $A_i\subset\partial\Sigma$ ($i\in[\N]$) such that the corresponding graph is $G$. Hence, the allowed set of entropies in such quantum gravity states equals the set of min-cut functions on weighted graphs.

The min-cut function is known to be a submodular function \cite{Headrick:2007km}, as expected on physical grounds, since it is a von Neumann entropy and hence must obey all the same properties. However, it further obeys a rich class of linear constraints that take the name of \emph{holographic entropy inequalities} (HEIs) \cite{Hayden:2011ag, Bao:2015bfa}. Understanding structural properties of such inequalities has been the focus of much work in recent years since it can provide hints on the nature of correlations in quantum gravity as well as the emergence of spacetime itself \cite{VanRaamsdonk:2010pw}. In this paper we prove several new theorems in this direction, motivated by the results of the paper \cite{Grimaldi:2025jad}.

\subsection{Summary of \cite{Grimaldi:2025jad}}
\label{ssec:summary}

In \cite{Grimaldi:2025jad}, the authors studied a particular class of HEIs that were named \emph{null reductions}. We first review what null reductions are, then motivate the physics behind them. Consider, as an example, the following HEI:
\begin{equation}\label{eq:mmi1}
    \ent{12} + \ent{23} + \ent{13} \geq \ent{1} + \ent{2} + \ent{3} + \ent{123}.
\end{equation}
This is known as \emph{monogamy of mutual information} (MMI), and was proven to be an HEI in \cite{Hayden:2011ag,Headrick:2013zda}. Now consider the following operation: remove all terms from the inequality that do not contain the index ``1":
\begin{equation}\label{eq:mmi-nr}
    \ent{12} + \ent{13} \geq \ent{1} + \ent{123}.
\end{equation}
This resulting inequality is called the null reduction of \eqref{eq:mmi1} on 1. Of course, the choice to null reduce on index 1 was arbitrary; one can reduce on any index. As a first observation, note that the resulting inequality is submodularity/SSA, and thus it is itself an HEI. This observation, along with its realization in thousands of cases obtained by null reducing other HEIs, motivated the following conjectures (the numbering matches that in \cite{Grimaldi:2025jad}):\footnote{More precisely, the conjectures of \cite{Grimaldi:2025jad} refer to \emph{superbalanced} inequalities, reviewed in  subsec.~\ref{sec:balance} (cf.~\cref{def:superbalance}), which covers all cases of interest.}
\begin{customconj}{3}\label{conj:3}
  \emph{If an inequality is an HEI, then all of its null reductions are HEIs.}  
\end{customconj}
\begin{customconj}{4}\label{conj:4}
\emph{If all the null reductions of an inequality are HEIs, then the inequality is an HEI.} 
\end{customconj}

A second observation made in \cite{Grimaldi:2025jad} was the following. From the null reduced inequality \eqref{eq:mmi-nr}, construct the following vectors ${x}$ and ${y}$:
\begin{equation}\label{eq:nr-vectors}
    x =         (\lambda_1 + \lambda_2,
        \lambda_1 + \lambda_3)\,,
\qquad 
    y = (       \lambda_1,
        \lambda_1 + \lambda_2 +  \lambda_3)\,,
\end{equation}
where $\lambda_1, \lambda_2, \lambda_3$ are arbitrary positive real variables. More generally, given a null-reduced inequality, construct a component of the vector $x$ by replacing a term $\ent{X}$ on the LHS with $\sum_{i \in X}\lambda_i$, and similarly for the vector $y$ but using the RHS terms. Then, note that for the vectors \eqref{eq:nr-vectors}, $x$ is \emph{majorized} by $y$ (or $x \preceq y$) for all positive $\lambda_1, \lambda_2, \lambda_3$. For general ${x}, {y} \in \mathbb{R}^d$, $x \preceq y$ if
\begin{equation}\label{eq:majdef}
 \sum_{n = 1}^{k} x_n^{\downarrow}\le \sum_{n = 1}^{k} y_n^{\downarrow}
  \quad( k = 1,2,\dots, d-1)\,,\qquad
  \sum_{n = 1}^{d} x_n=\sum_{n = 1}^{d} y_n \,,
\end{equation}
where $x_n^\downarrow$, $y_n^\downarrow$ are the $n$th largest components of $ x$, $y$ respectively. We refer the reader to Appendix A of \cite{Grimaldi:2025jad} for a brief review of majorization. Note also that this majorization property holds regardless of which index we null-reduce \eqref{eq:mmi1} on. This observation, together with further numerical investigations, motivated the following conjectures:
\begin{customconj}{1}\label{conj:1}
\emph{If an inequality is an HEI, then all of its null reductions obey the majorization property.}
\end{customconj}

\begin{customconj}{2}\label{conj:2}
\emph{If all the null reductions of an inequality obey the majorization property, the inequality is an HEI.} 
\end{customconj}

Conjecture 1 is related to the physical question of whether time-dependent states in quantum gravity obey the same constraints as static ones, which was the original motivation for the paper \cite{Grimaldi:2025jad}. Specifically, the null reduction of an HEI is obtained by considering a certain configuration of boundary regions residing on a common light cone whose vertex lies in the region being null reduced on. By design, for such a configuration the HEI is saturated in the vacuum state. One can then show that HEI is then safe against small perturbations of the state if and only if it obeys the majorization property.

\subsection{Results}
\label{ssec:results}

This paper's objective is two-fold. Firstly, we will address all of the above conjectures: namely, we will disprove conjectures \ref{conj:2} and \ref{conj:4} with direct counterexamples, and prove conjectures \ref{conj:1} and \ref{conj:3} under an assumption of contraction-provability.\footnote{\Cref{conj:3} was proved in \cite[Thm.\ 2]{Grimaldi:2025jad} under a different  assumption, namely that the inequality can be rendered in a tripartite form.} (The definition of contraction-provability can be found in subsec.\ \ref{ssec:matrix}.) As argued in \cite{Grimaldi:2025jad}, the proof of conjecture \ref{conj:1} provides strong new evidence for the statement that time-dependent and static states obey the same constraints.

Secondly, by borrowing tools from majorization theory, we will refine the setup and provide multiple new results. In particular, we will define several new combinatorial properties of inequalities, which we label as various forms of ``dominance'' and that are in some sense generalizations of the notions of balance and superbalance. We then establish the full set of logical relations connecting these notions with each other and with null reductions, the above majorization property, and contraction maps (used to prove that an inequality is an HEI). We provide not only proofs for all the logical implications but also, for each one-way implication, a counterexample for its converse. The map of these relations may be found in fig.\ \ref{fig:logic-map}. A key result, in addition to the resolutions of the above conjectures, is a necessary and sufficient combinatorial criterion for a null-reduced inequality to be a holographic entropy inequality.

The paper is structured as follows. In sec.\ \ref{sec:defs}, we continue our review of HEIs in more detail by providing basic definitions and theorems, and introducing the very useful language of $(0,1)$-matrices that will be employed in the rest of the paper. In sec.\ \ref{sec:dominance}, we introduce a new set of combinatorial properties of inequalities that are the basis for the main results of this paper, and in sec.\ \ref{sec:theorems} we provide proofs of all their logical relations. In sec.\ \ref{sec:discussion}, we provide an intuitive summary of the combinatorial properties and results discussed in the paper and present some interesting open questions. In appendix \ref{sec:contraction}, we discuss certain issues concerning the definition of contraction maps. Finally, in appendix \ref{sec:repurification}, we discuss the transformation of inequalities under choosing a different party to be the purifier, in particular explaining which combinatorial properties are preserved by this transformation.

\paragraph{Note added:} See also paper \cite{Czech:2026zca} for a discussion of similar issues, which will appear as a synchronized submission.

\section{Background}\label{sec:defs}

The goal of this section is to give a self-contained and brief overview of the setup, main definitions, and tools needed to understand the results of this paper.

\subsection{Notation \& conventions}

A matrix will be denoted by a bold capital letter like $\mathbf{A}$, its $j$'th column by $\mathbf{A}^{(j)}$ and its $i$'th row by $\mathbf{A}_{(i)}$. Given an $r\times c$ matrix $\mathbf{A}$ and subset $U \subseteq[c]:= \{1,2,\dots,c\}$, $\mathbf{A}^{(U)}$ will denote the submatrix induced by the columns with indices in $U$. A vector without a transpose sign ($v$) is a row vector, and with a transpose sign ($v^T$) is a column vector; $e$ denotes the row vector of all 1s, of the appropriate length given the context. Given vectors $v,w \in \mathbb{R}^m$ by $v \leq w$ we mean $v_i \leq w_i$ for all $i \in [m]$.

A bit-string is a vector with components in $\{0,1\}$. Given bit-strings $x,x'$ of the same length, their Hamming distance is the number of components where $x$ and $x'$ differ, which equals
\begin{equation}\label{eq:Haming-dot}
    |x-x'| =  |x| + |x'| - 2\,x\cdot x',
\end{equation}
where $|\cdot |$ is the $\ell^1$ norm 
($|v| = \sum_{i } |v_i|$) and $\cdot$ is the standard dot product. Also useful is the bitwise AND operator $\wedge$: $x \wedge x'$ is the bit-string with $1$ in every position where both $x$ and $x'$ are 1, and 0 otherwise; we will sometime call  $x \wedge x'$  the \emph{overlap} of $x$ and $x'$. One has $|x \wedge x' | = x\cdot x'$. Also, $x' \leq x$ is equivalent to the statement that the set of positions with 1s in $x'$ is a subset of the set of positions with 1s in $x$. For a set $U$, we will also use $|U|$ to denote the number of elements.

Vector majorization was defined in \eqref{eq:majdef}. For matrices, we will make use of the following notions.
\begin{defi}[PCM]\label{def:PCM}
    For $r\times c$ matrices $\mathbf{A}, \mathbf{B}$ we say that $\mathbf{A}$ is {\em positive-combinations majorized} by $\mathbf{B}$, written $\mathbf{A}\preceq^{\text{PC}} \mathbf{B}$, if $v\mathbf{A}\preceq v\mathbf{B}$  for all $v \in \mathbb{R}^r_{+}$.
\end{defi}

\noindent Other names for this inequality in the literature include {\em positive directional
majorization} and {\em price majorization}, see \cite[Definition IV.15.A.13]{MR2759813} and references therein.
If in the definition above we instead require $v$ to be any real vector, the corresponding notion is called {\em linear-combinations majorization} or {\em directional majorization}. This relation is strictly stronger, see \cite[Example IV.15.A.14]{MR2759813}. We can also go the other way and further reduce the combinations set.
\begin{defi}[BCM]\label{def:BCM}
    For $r\times c$ matrices $\mathbf{A}, \mathbf{B}$ we say that $\mathbf{A}$ is {\em binary-combinations majorized} by $\mathbf{B}$, written $\mathbf{A}\preceq^{\rm BC} \mathbf{B}$, if $v \mathbf{A} \preceq v\mathbf{B}$  for all $v \in \{0,1\}^r$.
\end{defi}

\subsection{Min-cut function \& entropy vector}\label{ssec:min-cuts}

Let $G$ be a weighted undirected graph with vertex set $\mathscr{V}$, partitioned into a set $\mathscr{X}$ of \emph{external} vertices and a set $\mathscr{I}$ of \emph{internal} vertices, and with edge set $\mathscr{E}$, with weights $w: \mathscr{E} \to \mathbb{R}^+$. Because of the quantum-information interpretation of these graphs, we will sometimes refer to external vertices as \emph{parties}.\footnote{It can be useful in certain settings, for example when considering the operations of merging and splitting parties, to define a ``party'' so that it can consist of multiple external vertices. In fact, this is crucial for e.g.\ representing certain important structures (such as the entropy cone extreme rays, mentioned below) by  tree graph models \cite{Hernandez-Cuenca:2022pst,Hubeny:2025bjo}. In this paper, we won't need this flexibility, so we will simply identify individual parties with individual external vertices.}

A \emph{cut} $V$ is a subset of $\mathscr{V}$, and the \emph{cut of edges} $\mathscr{C}(V)$ is the set of edges with one end in $V$ and one end in $V^c = \mathscr{V}\setminus V$, with total weight $|\mathscr{C}(V)|$ the sum of their weights
\begin{equation}
|\mathscr{C}(V)| = \sum_{e \in \mathscr{C}(V)} w(e).
\end{equation}
We define the \emph{min-cut function} (or \emph{holographic entanglement entropy}) $\mathsf{S}$ as the following function  on the set of subsets $X$ of $\mathscr{X}$:
\begin{equation}
    \ent{X} := \min_{V \sim X} |\mathscr{C}(V)|\,,
\end{equation}
where $V \sim X$ if $V  \cap \mathscr{X} = X$.

\begin{remark}\label{rem:min-cut-complement}
Since for any cut $\mathscr{C}(V)$ we have $|\mathscr{C}(V)| = |\mathscr{C}(\mathscr{V}\setminus V)|$, it follows that $\ent{X} = \ent{\mathscr{X}\setminus X}$. We can eliminate this redundancy by fixing a party $p\in\mathscr{X}$ called the \emph{purifier}, and considering only subsets $X\subseteq\mathscr{X}\setminus\{p\}$. Furthermore, since $\ent{\mathscr{X}} = 0$, one also has $\ent{\emptyset} = 0$, and one can simply restrict the function $\mathsf{S}$ to non-empty subsets of $\mathscr{X}\setminus\{p\}$.
\end{remark}

In the physics literature, the external vertices, or parties, are usually labelled by letters $A,B,\ldots$; here, we will instead adopt the graph-theory convention and label them by numbers $1,\ldots,\N+1$. Furthermore, unless explicitly indicated otherwise, we will choose the party labelled $\N+1$ to be the purifier $p$. With this convention, the min-cut function $\mathsf{S}$ is a function on $2^{[\N]}\setminus\{\emptyset\}$. We can collect its values into a vector,
\begin{equation}
    \vec{\mathsf{S}} \coloneq
    (\ent{1},\, \ent{2},\, \ent{3},\, \dots,\, \ent{12},\, \ent{13},\, \dots,\, \ent{12\cdots \N})
\end{equation}
which lives in $\mathbb{R}^{2^\N -1}$. In the above, juxtaposition of integers corresponds to the set of such integers (e.g. $\ent{12} := \ent{\{1,2\}}$).
In the physics literature $\vec{\mathsf{S}}$ is known as an \emph{entropy vector}, and the space $\mathbb{R}^{2^\N -1}$ it lives in as \emph{entropy space}.

\subsection{Holographic entropy inequalities \& the min-cut cone}
\label{ssec:HEC}

As we change the graph, including its weights, with $\N$ fixed, the entropy vector $\vec{\mathsf{S}}$ moves around in $\R^{2^\N-1}$. We are interested in the set $H_\N\subset \mathbb{R}^{2^\N -1}$ of entropy vectors for all weighted graphs with $\N+1$ external vertices. Multiplying all the weights of a given graph by a non-negative constant $c$ multiplies the corresponding entropy vector by $c$, $\vec{\mathsf{S}}\mapsto c\vec{\mathsf{S}}$. Furthermore, superposing two graphs (taking the union of the edge sets and internal vertex sets) adds their entropy vectors, $\vec{\mathsf{S}}=\vec{\mathsf{S}}_1+\vec{\mathsf{S}}_2$. It follows that $H_\N$ is a convex cone, called the \emph{min-cut cone} (in physics, the \emph{holographic entropy cone}, HEC). 

It can be shown that $H_\N$ is in fact a rational polyhedral cone, in other words it is defined by a finite number of inequalities with integer coefficients \cite{Bao:2015bfa, Avis:2021xnz}. Therefore, they can be written in the following form:
\begin{equation}\label{eq:ineq-form}
    \sum_{n = 1}^{\mL} \ent{X_n} \geq \sum_{n = 1}^{\mR} \ent{Y_n}\,,
\end{equation}
where the $X_n$ and $Y_n$ are non-empty subsets of $[\N]$. The same subset may appear multiple times on either the left-hand side or right-hand side of the inequality. We could equivalently require the subsets to be distinct and include positive integer coefficients in the inequality; however, the convention of having unit coefficients and repeated subsets will be more convenient for our analysis. 

\label{def:HEI} We will call an inequality of the form \eqref{eq:ineq-form} that holds for all vectors in $H_\N$ a \emph{holographic entropy inequality} (HEI).\footnote{In the paper \cite{Grimaldi:2025jad}, these inequalities were referred to as \emph{static holographic entropy inequalities} (sHEIs), to emphasize that by definition they hold for static holographic states. In this paper, we revert to the conventional term HEI. Note that in some of the earlier literature (such as 
\cite{Hubeny:2018trv,Hubeny:2018ijt,Hernandez-Cuenca:2022pst,Hernandez-Cuenca:2023iqh,Czech:2023xed,He:2024xzq}), 
the term HEI was reserved for the finite set of so-called \emph{primitive} inequalities (terminology explained in the next paragraph); here instead we use it to refer to the infinite set of all inequalities which arise from positive combinations of the primitive ones.} 
Note that we always write HEIs in the form $L\ge R$, with quantities referring to the left-hand side labelled by L and those referring to the right-hand side labelled by R.

An HEI that cannot be written as a positive integer linear combination of other HEIs is called a \emph{primitive} HEI. There are a finite number of them for each $\N$, and their saturation defines the facets of $H_\N$. This facet-based  description of $H_\N$ is known as the $H$-representation in polyhedral geometry, and it completely defines the cone. An equivalent complete description of $H_\N$ is provided by the set of its \emph{extreme rays}, known as the $V$-representation. The set of extreme rays generates $H_\N$ in the sense that their convex hull is $H_\N$. 

Relabelling the $\N$ non-purifier external vertices of a graph induces a linear action on the entropy space. Such relabellings define a representation of the symmetric group $S_\N$, under which $H_\N$ is necessarily invariant. In fact, it is invariant under the action of the larger permutation group $S_{\N+1}$ consisting of relabellings of all $\N+1$ external vertices, followed by replacing any subset that now includes the purifier by its complement. For example, if $\N=2$, under the permutation that exchanges 2 and 3, $\ent{2}$ becomes $\ent{3}=\ent{12}$, while $\ent{12}$ becomes $\ent{13}=\ent{2}.$

The following examples indicate some salient features of $H_\N$:
\begin{itemize}
\item $H_1$ is specified solely by the non-negativity of $\mathsf{S}$:
\begin{equation}
    \ent{1} \geq 0\,.
\end{equation}
\item $H_2$ is a simplicial cone with 3 facets. The inequalities defining its facets are \emph{subadditivity} (SA),
\begin{equation}\label{eq:SA}
    \ent{1} + \ent{2} \geq \ent{12}\,,
\end{equation}
and its images under the $S_3$ permutation group,
\be
\ent{1}+\ent{12}\ge \ent{2}\,,\qquad
\ent{2}+\ent{12}\ge \ent{1}\,.
\ee
We call the set of all inequalities obtained by permutations from a given one its \emph{orbit}, and it is conventional to only list one representative per orbit when discussing the extremal structure of $H_\N$. So, for $H_2$, the orbit of \eqref{eq:SA} has length 3. It lives in 3-dimensional space and has 3 extreme rays (which correspond to entropy vectors of Bell pairs).

\item For $\N = 3$ and $\N=4$, the cone is defined by two inequalities. These are SA and \emph{monogamy of mutual information} (MMI):
\begin{equation}\label{eq:MMI}
    \ent{12} + \ent{23} + \ent{13} \geq \ent{1} + \ent{2} + \ent{3} + \ent{123}\,,
\end{equation}
plus their permutations. For $\N=3$, orbits have lengths ${4\choose 2} =6$ and 1 respectively, so the cone is still simplicial.  For $\N=4$, this is no longer the case, since the orbit lengths are 10 and 10, leading to 20 facets (and 20 extreme rays), while it lives in 15-dimensional space. For $\N\le4$, all extreme rays are realized by $K_2$ (the graph with two vertices connected by one edge) and star graphs (graphs with a single internal vertex).

\item $H_5$ has 8 orbits, with lengths ranging from 10 to 90, for a total of 372 facets (and 2267 extreme rays organized into 19 orbits). Here the graphs realizing the given extreme rays are no 
longer realizable by star graphs; see fig.\ 1 in \cite{HernandezCuenca:2019wgh} or fig.\ 12 in \cite{Hernandez-Cuenca:2022pst}.

\item For $\N=6$, hundreds of primitive HEI orbits have been found \cite{Hernandez-Cuenca:2023iqh}; many more are being discovered and collected in the repository \cite{hecdatacurrent}, however, the full number is unknown and likely much larger.
\item Two infinite families of primitive HEIs, for arbitrarily large $\N$, have also been found \cite{Czech:2024rco}.
\end{itemize}

One can also study the dual cone,
\begin{equation}
    H^*_{\N} := \{ \vec{\mathsf{Q}}\in\mathbb{R}^{2^\N-1}:\forall \,\vec{\mathsf{S}} \in H_\N\,,\, \vec{\mathsf{Q}}\cdot \vec{\mathsf{S}} \geq 0 \}\,.
\end{equation}
The extreme rays and facets of $H_{\N}$ get mapped to the facets and extreme rays of $H^*_{\N}$ respectively. Any HEI can be written in the form $\vec{\mathsf{Q}}\cdot \vec{\mathsf{S}}\ge0$, where $\vec{\mathsf{Q}}$ is an integral point of $H_\N^*$. Given a vector $\vec{\mathsf{Q}}$, the quantity $\mathsf{Q}=\vec{\mathsf{Q}}\cdot \vec{\mathsf{S}}$ is called an \emph{information quantity} (IQ); if $\vec{\mathsf{Q}}\in H^*_\N$ then it is a \emph{holographic information quantity} (HIQ).

So far, we have been setting the purifier $p$ to be $\N+1$. It will give us a little more flexibility to allow an arbitrary choice of purifier $p\in[\N+1]$. We would then consider inequalities of the form \eqref{eq:ineq-form} where $X_n,Y_n$ are subsets of $[\N+1]\setminus\{p\}$. Of course, this is not really a generalization, since an inequality with $p\neq\N+1$ can be equivalently rewritten as one with $p=\N+1$ by replacing every $X_n,Y_n$ that contains $\N+1$ by its complement. We call an inequality with $p\neq\N+1$ an HEI if the one obtained in this way is an HEI.

\subsection{Matrix notation \& contraction maps}
\label{ssec:matrix}

Given an inequality of the form \eqref{eq:ineq-form} (with arbitrary purifier $p$), we can encode it into a pair of $(0,1)$ matrices $(\leftm,\rightm)$ with dimensions $(\N+1) \times \mL$ and $(\N+1) \times \mR$ respectively, with components defined as 
\begin{equation}
    L_{in} := \begin{cases}
    1 & \text{if } i \in X_n \\
    0 & \text{if } i \notin X_n
    \end{cases}, \qquad \quad R_{in} := \begin{cases}
    1 & \text{if } i \in Y_n \\
    0 & \text{if } i \notin Y_n
    \end{cases}\,.
\end{equation}
In other words, each column of $\leftm$ or $\rightm$ corresponds to one term on the LHS or RHS of \eqref{eq:ineq-form}, while each row corresponds to one external vertex.  Note that the matrices $\leftm$, $\rightm$ depend on the order in which the terms on the LHS and RHS respectively are written; changing the order permutes the columns. We will see that the various properties of matrix pairs that we will study are not affected by such permutations.

The rows $\leftocc_{(i)}$ and $\rightocc_{(i)}$ of the two matrices are known as \emph{occurrence vectors}. The occurrence vectors for the purifier $p$ are both necessarily 0. We let $\pairset$ be the set of pairs $(\leftm,\rightm)$ of $(0,1)$ matrices with $\N+1$ rows containing at least one common row of all 0s. We will generally treat an inequality and its matrix pair as interchangeable; for example, we will say that $(\leftm,\rightm)\in\pairset$ is an HEI if the corresponding inequality is an HEI.

\begin{eg}
For the MMI inequality \eqref{eq:MMI}, the $\leftm$ and $\rightm$ matrices are
\begin{align}
    &\,\,\,\,\,\ent{12} + \ent{23} + \ent{13} \quad\geq\quad \ent{1} + \ent{2} + \ent{3} + \ent{123}\\
    \leftm=&\begin{pmatrix}
    \quad 1 &\qquad\,\, 0 &\qquad\,\,\, 1\quad \\
    \quad 1 &\qquad\,\, 1 &\qquad\,\,\, 0\quad \\
    \quad 0 &\qquad\,\, 1 &\qquad\,\,\, 1\quad \\
    \quad 0 &\qquad\,\, 0 &\qquad\,\,\, 0\quad
    \end{pmatrix}\,\,\,,\,\,\,
    \begin{pmatrix}
   \quad 1 &\qquad 0 &\qquad 0 &\,\,\qquad 1\quad\\
   \quad 0 &\qquad 1 &\qquad 0 &\,\,\qquad 1\quad\\
   \quad 0 &\qquad 0 &\qquad 1 &\,\,\qquad 1\quad\\
   \quad 0 &\qquad 0 &\qquad 0 &\,\,\qquad 0\quad
    \end{pmatrix} = \rightm
\end{align}
\end{eg}

An essential tool in the study of HEIs is the contraction map.\footnote{Contraction maps were originally defined in terms of a weighted Hamming distance \cite{Bao:2015bfa}. However, this definition turns out to be too restrictive, as it is not satisfied for certain HEIs \cite{Avis:2021xnz}. The various types of contraction maps are discussed in detail in appendix \ref{sec:contraction}.}

\begin{defi}[Contraction map]
\label{def:contr_map} A \emph{contraction map} for a pair $(\leftm,\rightm) \in \pairset$ is a map $f: \{0,1\}^{\mL} \to \{0,1\}^{\mR}$ with the following properties:
\begin{enumerate}
    \item For all $i \in [\N + 1]$,  $f(\leftocc_{(i)}) = \rightocc_{(i)}$.
    \item For all $x,x'\in\{0,1\}^{\mL}$, $\ham{x}{x'} \geq  \ham{f(x)}{f(x')}$.
\end{enumerate}
\end{defi} 

It was shown in \cite{Bao:2015bfa} that if a contraction map exists then the corresponding inequality is an HEI. An argument was recently given for the converse \cite{Bao:2025sjn}; however we will not assume this. We will call an inequality admitting a contraction map a \emph{contraction-provable HEI} ({\CPI}). 
\label{def:cpHEI} 
To emphasize this distinction, we use HEI to refer to any inequality valid for the min-cut function on any weighted undirected graph with fixed number of terminals, and {\CPI} to refer to an HEI that can be proven using the contraction map method. 

\begin{eg}
The \emph{strong subadditivity} (SSA) inequality,
\begin{equation}\label{eq:SSA}
    \ent{12} + \ent{23} \geq \ent{2} + \ent{123}\,,
\end{equation}
corresponds to the following matrix pair:
\begin{equation}\label{eq:SSA-matrix}
    \leftm =
    \begin{pmatrix}
    1 & 0 \\
    1 & 1 \\
    0 & 1 \\
    0 & 0
    \end{pmatrix},\quad
    \rightm =
    \begin{pmatrix}
    0 & 1 \\
    1 & 1 \\
    0 & 1 \\
    0 & 0
    \end{pmatrix}.
\end{equation}
For this inequality, a contraction map $f$ is uniquely determined by the constraint $f(\leftocc_{(i)}) = \rightocc_{(i)}$:
\begin{equation}\label{eq:SSA-contraction}
    f(10) = 01\,,\quad f(11) = 11\,,\quad f(01) = 01\,,\quad f(00) = 00\,.
\end{equation}
As can be explicitly checked, this map also obeys the contraction property. SSA is therefore a \CPI.
\end{eg}

\subsection{Transformations of matrix pairs \& inequalities}
\label{sec:transformations}

As the reader may easily verify, the following transformations on the pair $(\leftm,\rightm)\in\pairset$, and the corresponding inequality, preserve the existence of a contraction map, as well as the PCM and BCM properties with one exception: duplicating rows doesn't preserve BCM.

\begin{enumerate}
\item{\textbf{Column permutations:}} Independently permuting the columns of $\leftm$ and of $\rightm$. This transformation corresponds to rearranging the terms on the two sides of the inequality. 

\item{\textbf{Row permutations:}} Applying the same permutation to the rows of $\leftm$ and $\rightm$. This transformation corresponds to relabelling the parties in the inequality. Consider the SSA inequality \eqref{eq:SSA-matrix} as an example; applying a $1\leftrightarrow 2$ row permutation on both $\leftm$ and $\rightm$ we get 
\begin{equation}
\leftm =
    \begin{pmatrix}
    1 & 0 \\
    1 & 1 \\
    0 & 1 \\
    0 & 0
    \end{pmatrix},\quad
    \rightm =
    \begin{pmatrix}
    0 & 1 \\
    1 & 1 \\
    0 & 1 \\
    0 & 0
    \end{pmatrix}\qquad \mapsto \qquad \leftm' =
    \begin{pmatrix}
    1 & 1 \\
    1 & 0 \\
    0 & 1 \\
    0 & 0
    \end{pmatrix},\quad
    \rightm' =
    \begin{pmatrix}
    1 & 1 \\
    0 & 1 \\
    0 & 1 \\
    0 & 0
    \end{pmatrix},
\end{equation}
which corresponds to transforming the original inequality \eqref{eq:SSA} to its permutation,
\begin{equation}
    \ent{12} + \ent{13} \geq \ent{1} + \ent{123}\,.
\end{equation}
\item \textbf{Duplicating rows:} Duplicating the same row on both $\leftm$ and $\rightm$ (or recombining duplicated rows on $\leftm$ and $\rightm$). This transformation corresponds to splitting a party into two parties (or the inverse operation). Of course, this changes the value of $\N$. Again, using \eqref{eq:SSA-matrix} as an example, by duplicating the first row,
\begin{equation}
\leftm =
    \begin{pmatrix}
    1 & 0 \\
    1 & 1 \\
    0 & 1 \\
    0 & 0
    \end{pmatrix},\quad
    \rightm =
    \begin{pmatrix}
    0 & 1 \\
    1 & 1 \\
    0 & 1 \\
    0 & 0
    \end{pmatrix} \qquad \mapsto \qquad \leftm' =
    \begin{pmatrix}
    1 & 0 \\
    1 & 0 \\
    1 & 1 \\
    0 & 1 \\
    0 & 0
    \end{pmatrix},\quad
    \rightm' =
    \begin{pmatrix}
    0 & 1 \\
    0 & 1 \\
    1 & 1 \\
    0 & 1 \\
    0 & 0
    \end{pmatrix},
\end{equation}
we have duplicated party 1 into parties 1 and 2, thereby embedding the inequality in the larger space for $\N = 4$, and obtaining the following instance of SSA with four parties: 
\begin{equation}
    \ent{123} + \ent{34} \geq \ent{3} + \ent{1234}\,.
\end{equation}
In contrast to PCM, duplicating rows does not preserve BCM. Indeed, a binary combination of the duplicated rows corresponds to a nonnegative integral combination of the original rows. Thus row duplication can enlarge the set tested by BCM from $(0,1)$-vectors to nonnegative integral vectors. (For an explicit example where duplicating rows does not preserve BCM, see the pair $(\leftm,\rightm)$ defined in \eqref{eq:BCMnotPCM}; duplicating the first row and triplicating the fourth and sixth rows yields a pair that violates BCM.)

\item \textbf{Appending columns:}
Appending the same column to both $\leftm$ and $\rightm$ (or deleting redundant columns). The column must have a 0 in the position of the purifier. This transformation corresponds to adding the same term on the LHS and RHS of the inequality (or canceling equal terms from both sides). Adding the term $\ent{3}$ to \eqref{eq:SSA} on both sides
\begin{equation}
    \ent{12} + \ent{23} + \ent{3} \geq \ent{2} + \ent{123} + \ent{3}
\end{equation}
corresponds to the following transformation on \eqref{eq:SSA-matrix}:
\begin{equation}
    \leftm =
    \begin{pmatrix}
    1 & 0 \\
    1 & 1 \\
    0 & 1 \\
    0 & 0
    \end{pmatrix},\quad
    \rightm =
    \begin{pmatrix}
    0 & 1 \\
    1 & 1 \\
    0 & 1 \\
    0 & 0
    \end{pmatrix} \qquad \mapsto \qquad \leftm' =
    \begin{pmatrix}
    1 & 0 & 0\\
    1 & 1 & 0\\
    0 & 1 & 1\\
    0 & 0 & 0
    \end{pmatrix},\quad
    \rightm' =
    \begin{pmatrix}
    0 & 1 & 0 \\
    1 & 1 & 0\\
    0 & 1 & 1\\
    0 & 0 & 0
    \end{pmatrix}\,.
\end{equation}
\item \textbf{Horizontal concatenation:} Concatenating $(\leftm,\rightm)$ horizontally with a pair $(\leftm',\rightm')\in\pairset$ that has a row of 0s in common with $(\leftm,\rightm)$ (i.e.\ concatenating each row of $\leftm$ with the corresponding row of $\leftm'$ to make a row of the new matrix, and similarly for $\rightm$, $\rightm'$). This corresponds to adding the corresponding inequalities. 
\end{enumerate}

Another kind of transformation on inequalities and matrix pairs that preserves the existence of a contraction map is choosing a different purifier. As mentioned in subsec.\ \ref{ssec:HEC}, this transformation also preserves the property of being an HEI. However, it interacts in non-trivial ways with the other combinatorial properties of inequalities that we will discuss in this paper, so we postpone a discussion of it to appendix \ref{sec:repurification}.

\subsection{Balanced \& superbalanced inequalities}
\label{sec:balance}

We now define two important properties of inequalities which will be used throughout the paper. Both are invariant under the transformations enumerated in the previous subsection.

\begin{defi}[Balance]\label{def:balance}
A pair $(\leftm,\rightm) \in \pairset$ is \emph{balanced} if, for all $i\in[\N+1]$, $|\leftocc_{(i)}| = |\rightocc_{(i)}|$. Correspondingly, an inequality is \emph{balanced} if every party $i\in[\N+1]$ appears the same number of times on the left- and right-hand sides. We define $\bset\subset\pairset$ to be the set of balanced pairs $(\leftm,\rightm)$.
\end{defi}

\begin{defi}[Superbalance]\label{def:superbalance}
A pair $(\leftm,\rightm) \in \pairset$ and the corresponding inequality are \emph{superbalanced} if, for all $i,j \in [\N + 1]$, $\leftocc_{(i)} \cdot \leftocc_{(j)} = \rightocc_{(i)} \cdot \rightocc_{(j)}$. Correspondingly, an inequality is \emph{superbalanced} if it is balanced and every pair $i,j\in[\N+1]$ appears the same number of times on the left- and right-hand sides. We define $\sbset\subset\bset$ to be the set of superbalanced pairs $(\leftm,\rightm)$.
\end{defi}

\begin{eg}
An example of a balanced but not superbalanced HEI is SSA as given in \eqref{eq:SSA}. This is evident from the fact that the doubleton 13 is missing from the LHS. On the other hand, the reader can check that MMI \eqref{eq:MMI} is superbalanced.
\end{eg}
Superbalance is a particularly important structural property of holographic entropy inequalities. In fact, every primitive HEI except SA is superbalanced \cite{He:2020xuo}. Thus, every HEI is a non-negative integer combination of SAs and a superbalanced HEI.

\subsection{Centered inequalities \& null reduction}
In this paper, we will be particularly interested in a class of inequalities that we call \emph{centered}. 

\begin{defi}[Centered]\label{def:centered}
Given $i\in[\N+1]$, a pair $(\leftm,\rightm) \in \pairset$ is \emph{centered on $i$} if it is balanced and if $\leftocc_{(i)} = e$ and $\rightocc_{(i)} = e$. Correspondingly, an inequality is \emph{centered on $i$} if it is balanced and the party $i$ appears in every term on both sides. We call $i$ the \emph{central party}. We say  $(\leftm,\rightm)$ and the corresponding inequality are \emph{centered} if they are centered on some party. We define $\cset \subset \bset$ to be the set of centered pairs $(\leftm,\rightm)$.
\end{defi}

\begin{remark}
It follows directly from the definition that any centered inequality has an equal number of terms on the LHS and RHS, i.e.\ $\mL = \mR$. It also follows that the central party cannot be the purifier.
\end{remark}

\begin{eg}
SSA as written in \eqref{eq:SSA} is an HEI centered on party 2. Of course, there are plenty of examples of centered inequalities that are not HEIs, such as
\begin{equation}
    \ent{13} + \ent{124} \ge \ent{14} + \ent{123}.
\end{equation}
The fact that this is not an HEI follows immediately from the observation that under permuting 3 and 4, the LHS and RHS get flipped.
\end{eg}

\begin{remark}\label{rem:transformationscentered}
The transformations listed in subsec.\ \ref{sec:transformations} preserve the property of being centered on party $i$, with the following caveats:
\begin{itemize}
\item For transformation 2 (row permutation), if the permutation maps party $i$ to party $j$, then the resulting inequality is centered on $j$.
\item Transformation 4 (appending columns) only preserves being centered on $i$ if the appended column has a 1 in the row $i$.
\item For transformation 5 (horizontal concatenation), if both $(\leftm,\rightm)$ and $(\leftm',\rightm')$ are centered on $i$, then the concatenated pair is centered on $i$.
\end{itemize}
\end{remark} 

A centered inequality can be obtained from any superbalanced inequality by the operation of null reduction, whose physical importance was explained in \cite{Grimaldi:2025jad} and contextualized in the Introduction.

\begin{defi}[Null reduction]\label{def:nr}
Given a pair $(\leftm, \rightm) \in \pairset$, we say $i \in [\N+1]$ is \emph{present} if $\leftm_{(i)}$ and $\rightm_{(i)}$ are not both 0, or equivalently if $i$ appears in the corresponding inequality. The \emph{null reduction} of $(\leftm, \rightm)$ on $i$ is then the pair $(\leftm', \rightm')$ obtained by keeping all columns of $(\leftm, \rightm)$ containing a 1 in row $i$. Correspondingly, the \emph{null reduction} of an inequality on $i$ is the inequality obtained by keeping all terms containing $i$.\footnote{We require $i$ to appear in the inequality since otherwise null reduction would lead to the trivial inequality $0\ge0$, an edge case we would rather not deal with. In particular, null reducting on the purifier is forbidden.}
\end{defi}

It was proved in \cite[Lemma 1]{Grimaldi:2025jad}
that the null reduction of a superbalanced inequality is always centered. For ease of reference, we reproduce the proof here, rephrasing it in the matrix language. 

\begin{lemma}\label{lem:nr-is-centered}
If $(\leftm,\rightm) \in \sbset$, then its null reduction $(\leftm',\rightm')$ on any present $i\in[\N+1]$ is centered on $i$.  
\end{lemma}
\begin{proof}
We need to show that the null reduction of $(\leftm,\rightm)$ on party $i$ gives rise to a balanced pair $(\leftm',\rightm')$  with $\leftocc'_{(i)} = \rightocc'_{(i)} = e$. Clearly $\leftocc'_{(i)} = \rightocc'_{(i)} = e$, since by the definition of null reduction, we keep all columns of $(\leftm,\rightm)$ that have a 1 on the $i^{\text{th}}$ row. By superbalance of $(\leftm,\rightm)$, for any $j\in [\N+1]$, we have
\begin{equation}   
|\leftocc'_{(j)}|=\leftocc'_{(i)} \cdot \leftocc'_{(j)}=\leftocc_{(i)} \cdot \leftocc_{(j)} = \rightocc_{(i)} \cdot \rightocc_{(j)} = \rightocc'_{(i)} \cdot \rightocc'_{(j)} = |\rightocc'_{(j)}|\,,
\end{equation}
which proves balance.
\end{proof}

\section{Dominance properties}\label{sec:dominance}

Here we introduce several key properties of inequalities (or equivalently matrix pairs). In the next section, we will prove logical relations among them. All the properties presented below are preserved by the transformations on the inequality (or matrix pair) of subsec.\ \ref{sec:transformations}. Some further intuitive descriptions of these dominance properties can be found in subsec.\ \ref{ssec:discussion-summary}. While our main interest will be in centered inequalities, to which most of the theorems of sec.~\ref{sec:theorems} pertain, the definitions can be applied to general inequalities.

\subsection{Dominance}

\begin{defi}[Dominance map]\label{def:dominance-map}
A \emph{dominance map} for a pair $(\leftm,\rightm)\in\pairset$ is a map $f: \{0,1\}^{\mL} \to \{0,1\}^{\mR}$ such that, for all $u\in\{0,1\}^{\mL}$, $|u|=|f(u)|$ and $\leftm u^T\le\rightm f(u)^T$.
\end{defi}
\begin{defi}[Dominance]\label{def:dominance}
$(\leftm,\rightm) \in \pairset$ obeys \emph{dominance} if it admits a dominance map.
\end{defi}

\begin{remark}
The above definition has a dual combinatorial description in terms of subsets of terms of the inequality represented by $(\leftm, \rightm)$. Indeed, a bit-string $u \in \{0,1\}^{\mL}$ can be understood as dictating whether each column is to be included (1) or to be excluded (0). For example,  the bit-string $u = (1,0,1)$ corresponds to keeping columns 1 and 3 only, i.e.\ to the subset $U = \{1,3\}$. Therefore, there is a bijection between subsets $U\subseteq[\mL]$ and bit-strings $u \in \{0,1\}^{\mL}$. Then, the existence of a dominance map $f(u)= v$ can be rephrased as the existence, for each $U\subseteq[\mL]$, of a subset $V\subseteq [\mR]$ with the same number of elements as $U$ and such that 
\begin{equation}\label{eq:UVdom}
    \leftm^{(U)}\, e^T \leq \rightm^{(V)}\, e^T\,.
\end{equation}
We can equivalently express the $i$th row of \eqref{eq:UVdom} as
\begin{equation}\label{eq:UVdomi}
    u\cdot \leftocc_{(i)} \leq f(u) \cdot \rightocc_{(i)}\,.
\end{equation}
In terms of the inequality, the statement is that for any subset $U$ of the LHS terms there exists an equal-sized subset $V$ of the RHS terms containing every party at least as many times as $U$;
we say that $U$ \emph{is dominated by} $V$.
\end{remark}

Let's see an example and a non-example of dominance.

\begin{eg}[Example of dominance]\label{eg:example-dominance}
Consider SSA
\begin{equation}
    \ent{12} + \ent{23} \ge \ent{2} + \ent{123}\,.
\end{equation}
Then the possible $U$-subsets of the LHS are $\{1\},\{2\}$ and $\{1,2\}$. The corresponding dominating $V$-subsets are $\{2\}, \{2\}$ and $\{1,2\}$, i.e.
\begin{align}
    &U = \{1\}& &V = \{2\}  &&\{\ent{12}\} \text{ dominated by } \{\ent{123}\}\\
    &U = \{2\} &&V = \{2\}  &&\{\ent{23}\} \text{  dominated by } \{\ent{123}\}\\
    &U = \{1,2\}& &V = \{1,2\}  &&\{\ent{12}, \ent{23}\} \text{ dominated by } \{\ent{2},\ent{123}\}
\end{align}
\end{eg}

\begin{eg}[Non-example of dominance]\label{eg:non-example-dominance}
Consider the following $\N = 6$ HEI\footnote{This primitive inequality is a permutation of the 14th in the $\N=6$ list in the publicly available database \cite{hecdata}.}, which we write in matrix form to avoid clutter and to provide a practical example of working with matrices:
\begin{equation}
    \leftm = \begin{pmatrix}
1&0&1&0&0&0&0\\
0&1&1&0&0&0&1\\
1&1&0&1&1&0&0\\
1&0&1&1&0&1&1\\
0&1&1&1&1&1&0\\
0&0&0&0&1&1&0\\
0&0&0&0&0&0&0
\end{pmatrix},\quad \rightm= \begin{pmatrix}
0 & 0 & 0 & 0 & 1 & 0 & 0 & 0 & 1\\
0 & 0 & 1 & 0 & 0 & 0 & 1 & 0 & 1\\
1 & 0 & 0 & 1 & 0 & 0 & 0 & 1 & 1\\
0 & 1 & 0 & 0 & 1 & 0 & 1 & 1 & 1\\
0 & 0 & 0 & 1 & 0 & 1 & 1 & 1 & 1\\
0 & 0 & 0 & 0 & 0 & 1 & 0 & 1 & 0\\
0 & 0 & 0 & 0 & 0 & 0 & 0 & 0 & 0
\end{pmatrix}\,.
\end{equation}
For $U = \{1,2,3\}$, there is no dominating $V$. Indeed, to satisfy \eqref{eq:UVdom} for $i=1$, we need $V \supset \{5,9\}$; for $i=2$, we additionally need to include 3 or 7, but that already saturates $|V|$, so dominance then fails on $i=3$. This example shows that not every HEI obeys dominance.
\end{eg}

\subsection{Inclusion dominance}

Dominance is quite a restrictive property, due to the equal-size condition on the subsets $U$ and $V$. Clearly if there was no such condition, \eqref{eq:UVdom} would be trivially satisfied by any balanced inequality. However, we will show that all centered HEIs not only obey dominance but also a stronger property, namely the existence of a dominance map that preserves inclusion relations: for every pair $U'$, $U$ of LHS subsets with $U'\subseteq U$, there exist dominating subsets $V'$, $V$ of RHS subsets with $V'\subseteq V$. We refer to this stronger property as inclusion dominance.

\begin{defi}[Inclusion dominance map]\label{def:inclusion-dominance-map}
    An \emph{inclusion dominance map} for a pair $(\leftm,\rightm) \in \pairset$ is a dominance map $f: \{0,1\}^{\mL} \to \{0,1\}^{\mR}$ such that, for all $ u,u' \in \{0,1\}^{\mL}$ with  $u' \leq u$, $f(u') \leq f(u)$.
\end{defi}
\begin{defi}[Inclusion dominance]\label{def:inclusion-dominance}
We say $(\leftm,\rightm) \in \pairset$ obeys \emph{inclusion dominance} if it admits an inclusion dominance map.
\end{defi}

Let's see an example and a non-example.

\begin{eg}[Example of inclusion dominance]
Consider the following $\N=5$ centered HEI\footnote{This inequality is the null reduction on party 1 of the 6th HEI in the $\N=6$ list in \cite{hecdata} (where here we have joined parties 5 and 6 to make it an $\N = 5$ inequality).}: 
\begin{equation}
    \ent{123} + \ent{124} + \ent{125} + \ent{135} \geq \ent{12} + \ent{13} + \ent{1235} + \ent{1245}\,.
\end{equation}
We are not going to prove that inclusion dominance holds for all possible subsets, but focusing at least on the chain of inclusions $\{1\} \subset \textcolor{darkgreen}{\{1,2\}} \subset \textcolor{blue}{\{1,2,3\}} \subset \textcolor{red}{\{1,2,3,4\}}$ for the $U$-subsets we can see that the corresponding dominating subsets $V$ preserve the inclusion structure, as follows: 
\begin{equation*}
    \fcolorbox{red}{red!15}{\fcolorbox{blue}{blue!15}{\fcolorbox{darkgreen}{darkgreen!15}{\fcolorbox{black}{gray!15}{$\ent{123}$} + $\ent{124}$} + $\ent{125}$} + $\ent{135}$} \geq \fcolorbox{red}{red!15}{\fcolorbox{blue}{blue!15}{$\ent{12}$} + $\ent{13}$ + \fcolorbox{blue}{blue!15}{\fcolorbox{darkgreen}{darkgreen!15}{\fcolorbox{black}{gray!15}{$\ent{1235}$} + $\ent{1245}$}}}
\end{equation*}
Indeed we have $\{3\} \subset \textcolor{darkgreen}{\{3,4\}} \subset \textcolor{blue}{\{1,3,4\}} \subset \textcolor{red}{\{1,2,3,4\}}$ on the RHS.
\end{eg}
\begin{eg}[Non-example of inclusion dominance]\label{eg:non-example-inclusion-dom}
Consider the following centered $\N=5$ non-HEI inequality:\footnote{The inequality is the null reduction on party 1 of the non-HEI obtained from the difference between MMI on parties $\{1\}, \{2,5\}$ and $\{3,4\}$ and the HEI
\begin{multline}
    \ent{24}+\ent{35}+\ent{123}+\ent{124}+\ent{125}+\ent{134}+\ent{145}\\ \geq \ent{2}+\ent{3}+\ent{4}+\ent{5}+\ent{12}+\ent{14}+
    \ent{135}+\ent{1234}+\ent{1245},
\end{multline}
which is a permutation of the 4th inequality in the $\N = 5$ list in \cite{hecdata}. The null reduction on party 1 of this difference is \eqref{eq:eg-ne-id} which is also not an HEI, as explained in \ref{ssec:cond_incdom}.
}
\begin{equation}\label{eq:eg-ne-id}
    \ent{12} + \ent{14} + \ent{135} + \ent{1234} +\ent{1245} \geq \ent{1} + \ent{123} + \ent{124} + \ent{145} + \ent{12345}\,.
\end{equation}
It can be shown that this inequality obeys dominance. However, it fails inclusion dominance: for $U = \textcolor{red}{\{3,4,5\}}$ there is a unique dominating subset $V = \textcolor{red}{\{2,4,5\}}$, but for $U'=\textcolor{blue}{\{4,5\}}\subset U$ the unique dominating subset is $V' = \textcolor{blue}{\{3,5\}} \not\subset V$, i.e.\
\begin{multline}
    \ent{12} + \ent{14} + \fcolorbox{red}{red!15}{$\ent{135}$ + \fcolorbox{blue}{blue!15}{$\ent{1234} +\ent{1245}$}} \\
    \geq \ent{1} + \fcolorbox{red}{red!15}{$\ent{123}$} + \fcolorbox{blue}{blue!15}{$\ent{124}$} + \fcolorbox{red}{red!15}{$\ent{145}$ + \fcolorbox{blue}{blue!15}{$\ent{12345}$}}\,.
\end{multline}
\end{eg}

In the next section we will show that inclusion dominance is a necessary and sufficient condition for a centered inequality to be a \CPI, whereas dominance is only necessary. Inclusion dominance clearly imposes a very strict combinatorial constraint on the structure of centered {\CPI}s. In contrast, the weaker notion of dominance, while only necessary for a centered inequality to be a \CPI, is nonetheless important because, as we will see, it is a  sufficient condition for PCM (cf.~Def.\ \ref{def:PCM}), $\leftm \preceq^{\rm PC} \rightm$. We now introduce a necessary combinatorial condition for PCM, which we call region dominance.

\subsection{Region dominance}

\begin{defi}[Region dominance]\label{def:region-dominance}
We say that $(\leftm,\rightm)\in \pairset$ obeys \emph{region dominance} if it is balanced and, for any collection of parties $X \subseteq [\N+1]$, we have
\begin{equation}\label{eq:region-dominance}
     \Big| \bigwedge_{i \in X} \leftocc_{(i)} \Big| \leq \Big| \bigwedge_{i \in X} \rightocc_{(i)} \Big|\,;
\end{equation}
equivalently, the number of terms on the LHS of the inequality containing $X$ does not exceed the number on the RHS.
\end{defi}

Region dominance can be understood as the generalization of the notions of balance and superbalance to collections of parties of any size, but where equality between the LHS and RHS is weakened to an inequality. The origin of the property’s name lies in the physics convention that a region denotes a collection of parties. 

\begin{eg}[Example of region dominance]\label{eg:eg-region-dominance}
Consider the following $\N = 6$ non-HEI:\footnote{This inequality is not an HEI because it is the difference of two primitive HEIs. See subsec.\ \ref{ssec:bhei-rd} for more details on this.}
\begin{multline}\label{eq:RDexample}
    \ent{34} + \ent{123} + \ent{124} + \ent{136} + \ent{145} + \ent{236} + \ent{245} \geq\\
    \ent{2} + \ent{13} + \ent{14} + \ent{36} + \ent{45} + \ent{234} + \ent{1236} + \ent{1245}\, . 
\end{multline}
We will check with some examples how the above inequality satisfies region dominance. For example consider the regions $\textcolor{red}{\{5\}},\textcolor{blue}{\{1,2,3\}}, \textcolor{darkgreen}{\{2,3,4\}}$ and $\textcolor{orange}{\{1,2,4,5\}}$, i.e.\
\begin{multline}
    \ent{34} + \ent{\fcolorbox{blue}{blue!30}{123}} + \ent{124} + \ent{136} + \ent{14\fcolorbox{red}{red!30}{5}} + \ent{236} + \ent{24\fcolorbox{red}{red!30}{5}} \geq\\
    \ent{2} + \ent{13} + \ent{14} + \ent{36} + \ent{4\fcolorbox{red}{red!30}{5}} + \ent{\fcolorbox{darkgreen}{darkgreen!30}{234}} + \ent{\fcolorbox{blue}{blue!30}{123}6} + \ent{\fcolorbox{orange}{orange!20}{124\fcolorbox{red}{red!30}{5}}}\,. 
\end{multline}
Indeed, each subset appears at least as many times on the RHS as on the LHS. (We will see in subsec.\ \ref{ssec:bhei-rd} that \eqref{eq:RDexample} is not an HEI.)
\end{eg}

Notice the key distinction between region dominance and dominance:  while dominance concerns subsets of terms on each side of an inequality which correspond to different subsystems, region dominance concerns subsets of parties.  In this sense, both conditions in fact constitute different generalizations of balance, for which  $\le$ is specialized to $=$ and a restriction is made on the selection: for region dominance it comprises just the single parties $|X|=1$, whereas for dominance it comprises the full subset of terms, $|u|=\mL = \mR$. Inclusion dominance, one the other hand, has no direct relation to balance, since it concerns incomplete collections of terms. Starting from the full collection of terms, it states that for any choice of term to drop from the LHS collection, there is a ``smaller'' term we can drop from the RHS collection, and that we can continue this chain all the way down to a single term, where the RHS one is at least as large as the LHS one. (This in particular implies that the RHS must contain both the largest and the smallest term appearing in the inequality.)

\section{Theorems \& counterexamples}\label{sec:theorems}

This section will be dedicated to proving the main theorems of this paper. To help the reader follow the logical relations among the properties introduced in the previous section, in fig.\ \ref{fig:logic-map} we provide a diagram of the various implications. For every one-way implication, we will also provide a counterexample to the converse implication in its respective subsection.

Theorem \ref{thm:if-HEI-then-contraction} establishes Conjecture \ref{conj:3} for {\CPI}s, while Theorems  \ref{thm:if-HEI-then-contraction}, \ref{thm:contraction-equals-inclusion-dom}, \ref{thm:if-dominance-then-PCM} together establish Conjecture \ref{conj:1} for {\CPI}s. The counterexample we will display for the converse to Theorem \ref{thm:if-HEI-then-contraction} falsifies Conjecture \ref{conj:4}. Using Theorems \ref{thm:contraction-equals-inclusion-dom} and \ref{thm:if-dominance-then-PCM}, the same counterexample also falsifies Conjecture \ref{conj:2}.

\begin{figure}
\begin{center}
\begin{tikzpicture}[
    font=\footnotesize,
    box/.style={
        draw, fill=blue!10, inner sep=4pt,align=center,
        fill=#1!15, thick
    }
]
\pgfmathsetmacro{\xT}{0.5};   
\pgfmathsetmacro{\xS}{-0.6};  
\pgfmathsetmacro{\xA}{3.0};   
\pgfmathsetmacro{\yA}{1.7};   

\node at (-5,1) {\hyperref[def:superbalance]{{\Large \textcolor{red}{Superbalanced}}}};
\node[box=red]  (Scm) at (-5,0) {\hyperref[def:cpHEI]{{\CPI}}};

\node at (0,1) {\hyperref[def:centered]{{\Large \textcolor{orange}{Centered}}}};
\node at (-2.5,2.5) {\large \hyperref[def:nr]{Null reduce}};
\node[box=orange] (Pcm) at (0,0) {\hyperref[def:cpHEI]{{\CPI}}};
\node[box=orange] (Pid)  at (0,-1*\yA) {\hyperref[def:inclusion-dominance]{Inclusion dominance}};
\node[box=orange] (Pd)   at (0,-2*\yA) {\hyperref[def:dominance]{Dominance}};
\node[box=orange] (Ppcm) at (0,-3*\yA) {\hyperref[def:PCM]{PCM}};
\node[box=orange] (Pbcm) at (0,-4*\yA) {\hyperref[def:BCM]{BCM}};
\node[box=orange] (Prd)  at (0,-5*\yA) {\hyperref[def:region-dominance]{Region dominance}};

\node[box=darkgreen] (Bhei) at (5,0) {\hyperref[def:HEI]{HEI}};
\node[box=darkgreen] (Brd)  at (5,-\yA) {\hyperref[def:region-dominance]{Region dominance}};
\node at (5,1) {\hyperref[def:balance]{{\Large \textcolor{darkgreen}{Balanced}}}};

\draw[-implies, double distance=2.5pt,shorten <=1mm,
  shorten >=1mm, thick] (Scm) -- (Pcm);
\draw[implies-implies, double distance=2.5pt,shorten <=1mm,
  shorten >=1mm, thick] (Pcm) -- (Pid);
\draw[-implies, double distance=2.5pt,shorten <=1mm,
  shorten >=1mm, thick] (Pid) -- (Pd);
\draw[-implies, double distance=2.5pt, shorten <=1mm,
  shorten >=1mm, thick] (Pd) -- (Ppcm);
\draw[-implies, double distance=2.5pt, shorten <=1mm,
  shorten >=1mm, thick] (Ppcm) -- (Pbcm);
\draw[-implies, double distance=2.5pt, shorten <=1mm,
  shorten >=1mm, thick] (Pbcm) -- (Prd);
\draw[-implies, double distance=2.5pt, shorten <=1mm,
  shorten >=1mm, thick] (Bhei) -- (Brd);
\draw[->] (-5,1.5) to[bend left=25] (0,1.5);
\draw[->, dashed] (Scm) to[bend right=15] (-0.6,-5.1);

\node at (1.2*\xT, -0.5*\yA) {\,\, Thm.\ \ref{thm:contraction-equals-inclusion-dom}};
\node at (1.2*\xT, -1.5*\yA) {\,\, by def.
};
\node at (1.2*\xT, -2.5*\yA) {\,\, Thm.\ \ref{thm:if-dominance-then-PCM}};
\node at (1.2*\xT, -3.5*\yA) {\,\, by def.};
\node at (1.2*\xT, -4.5*\yA) {\,\, Thm.\ \ref{thm:if-bcm-then-rd}};
\node at (-2.5, 0.375*\yA)     {\,\, Thm.\ \ref{thm:if-HEI-then-contraction}};
\node at (-2.5, 0.175*\yA)     {Conj.\ \ref{conj:3} \checkmark};

\node at (5+1.2*\xT,  -0.5*\yA) {\,\, Thm.\ \ref{thm:if-balanced-hei-then-region-dominance}};

\node at (\xS, -0.5*\yA) {\textsection\ref{ssec:cond_incdom}};
\node at (\xS, -1.5*\yA) {\textsection\ref{ssec:cond_incdom}};
\node at (\xS, -2.5*\yA) {\textsection\ref{ssec:dom_PCM}};
\node at (\xS, -3.5*\yA) {\textsection\ref{ssec:pcm-bcm}};
\node at (\xS, -4.5*\yA) {\textsection\ref{ssec:bcm-rd}};
\node at (-2.5, -\yA/4)   {\textsection\ref{ssec:contS_contNR}};
\node at (5+\xS, -0.5*\yA) {\textsection\ref{ssec:bhei-rd}};

\node at (-3.75, -2.1*\yA) {Conj.\ \ref{conj:1} \checkmark};

\end{tikzpicture}
\end{center}
\caption{Map of the logical implications among the theorems presented in this paper. We organized the map based on the different type of inequalities the theorems apply to: superbalanced ($\sbset$), centered ($\cset$), and balanced ($\bset$). We hyperlinked every logical implication to the respective theorem and section, as well as the names in the boxes to their respective definitions. We also show how conjectures \ref{conj:1} and \ref{conj:3} are resolved for {\CPI}s  (conjectures \ref{conj:2} and \ref{conj:4} go in the opposite directions, and we provide counterexamples in subsec.\ \ref{ssec:dom_PCM} and \ref{sssec:converses} respectively).}
\label{fig:logic-map}
\end{figure}

\subsection{Null reduction maps {\CPI}s to {\CPI}s}
\label{ssec:contS_contNR}

In this subsection, we prove that every null reduction of a superbalanced \CPI\ is a \CPI.

\begin{thm}\label{thm:if-HEI-then-contraction}
If $(\leftm,\rightm)\in \sbset$ is a \CPI, then its null reduction on any present $i\in[\N+1]$ is a \CPI.
\end{thm}

\begin{proof}
Call the null-reduced pair $(\leftm',\rightm')$. Let $ m:=|\leftocc_{(i)}| = |\rightocc_{(i)}|$, and without loss of generality re-arrange the columns of $\leftm$ and $\rightm$ such that the first $m$ columns contain a 1 in row $i$ (i.e.\ we reorder the terms in the inequality so that the first $m$ terms on each side contain party $i$). Decompose the occurrence vector for party $j$ as $\leftocc_{(j)}=(\leftocc'_{(j)},\bar{\leftocc}_{(j)})$, where $\leftocc'_{(j)}$ is the first $m$ components and $\bar \leftocc_{(j)}$ is the remaining $\mL-m$  components, and similarly $\rightocc_{(j)}=(\rightocc'_{(j)},\bar \rightocc_{(j)})$. $\leftm'_{(j)}$, $\rightm'_{(j)}$ are the occurrence vectors for the null-reduced pair $(\leftm',\rightm')$. By Lemma \ref{lem:nr-is-centered}, $(\leftm',\rightm')$ is balanced; together with the fact that $(\leftm,\rightm)$ is also balanced, we have
\be\label{sbhat}
|\leftocc'_{(j)}| = |\rightocc'_{(j)}|, \quad |\bar \leftocc_{(j)}|=|\bar \rightocc_{(j)}|\,.
\ee

Next, we will show that the contraction map $f$ for $(\leftm,\rightm)$ maps bit strings of the form $(x',0)$ to bit strings of the form $(y',0)$ (where $x',y'\in\{0,1\}^m$, $(x',0)$ is the bit string of length $\mL$ padded with $\mL-m$ zeros, and similarly for $(y',0)$). To show this, we decompose $f((x',0))=(y',\bar y)$. The contraction property for $f$ implies
\begin{align}
|x'|= \ham{(x',0)}{\leftocc_{(p)}} &\ge \ham{(y',\bar y)}{\rightocc_{(p)}} =|y'|+|\bar y|\\
m-|x'|=\ham{(x',0)}{\leftocc_{(i)}}&\ge\ham{(y',\bar y)}{\rightocc_{(i)}}=m-|y'|+|\bar y|
\end{align}
(where $p$ is the purifier). Adding the inequalities gives $m\ge m+2|\bar y|$. Since $|\bar y|\ge0$, we must have $|\bar y|=0$, hence $\bar y=0$.

Given the result of the previous paragraph, let $f':\{0,1\}^m\to\{0,1\}^m$ be the map that takes $x'$ to $y'$. We will now show that $f'$ is a contraction map: it has the contraction property and maps $\leftm'_{(j)}$ to $\rightm'_{(j)}$. For any $x',z' \in \{0,1\}^m$, we have
\be
|x'-z'|=|(x',0)-(z',0)|\ge |f((x',0))-f((z',0))|=|f'(x')- f'(z')|\,,
\ee
so $f'$ has the contraction property. For any $j\in[\N+1]$, we have
\begin{align}\label{eq:Lbarj}
    |\bar\leftocc_{(j)}|
    &= \ham{(\leftocc'_{(j)},0)}{\leftocc_{(j)}}\nonumber \\
    &\ge
    \ham{ f(\leftocc'_{(j)},0)}{f(\leftocc_{(j)}) } =
\ham{(f'(\leftocc'_{(j)}),0)}{\rightocc_{(j)}}=\ham{f'(\leftocc'_{(j)})}{\rightocc'_{(j)}}+|\bar \rightocc_{(j)}|\,.
\end{align}
Together, \eqref{sbhat} and \eqref{eq:Lbarj} imply $\ham{ f'(\leftocc'_{(j)})}{\rightocc'_{(j)}} \le 0$, which in turn implies $\ham{f'(\leftocc'_{(j)})}{\rightocc'_{(j)}}=0$, hence $f'(\leftocc'_{(j)})=\rightocc'_{(j)}$.
\end{proof}
\addtocontents{toc}{\protect\setcounter{tocdepth}{2}}
\subsubsection{Converse statements}
\label{sssec:converses}

Theorem \ref{thm:if-HEI-then-contraction} has the following corollary:
\begin{cor}
If $(\leftm,\rightm)\in \sbset$ is a \CPI, then its null reduction on any present $i\in[\N+1]$ is an HEI.
\end{cor}
\noindent The converse of this corollary is the statement that, given a pair $(\leftm,\rightm)\in \sbset$, if its null reductions on all present $i\in[\N+1]$ are HEIs, then it is a \CPI. In turn this statement would imply Conjecture \ref{conj:4}. We now show that Conjecture \ref{conj:4}, and therefore this statement, are false by direct counterexample. (The counterexample will also disprove the converse of Theorem \ref{thm:if-HEI-then-contraction} itself.)

Consider the following non-HEI\footnote{The inequality is not an HEI because it is the difference of two primitive $\N = 5$ HEIs, specifically, the 6th inequality in the $\N = 5$ list of \cite{hecdata} minus a permutation of the 7th.} superbalanced inequality:
\begin{multline}\label{eq:ce-c4}
    \ent{123} + \ent{123} + \ent{124} + \ent{125} + \ent{135} + \ent{135} + \ent{234} + \ent{235} \geq\\
    \ent{12} + \ent{13} + \ent{15} + \ent{23} + \ent{24} + \ent{35} + \ent{1234} + \ent{1235} + \ent{1235}\,.
\end{multline}
The null reduction of the above on party 1 is
\begin{multline}
    \ent{123} + \ent{123} + \ent{124} + \ent{125} + \ent{135} + \ent{135} \geq\\
    \ent{12} + \ent{13} + \ent{15} +  \ent{1234} + \ent{1235} + \ent{1235}\,,
\end{multline}
which is just the sum of three instances of SSA: one on subsets $\{123\}$ and $\{124\}$, one on subsets $\{123\}$ and $\{135\}$, and one on subsets  $\{125\}$ and $\{135\}$. The null reductions on parties 2 and 3 are also the sum of three instances of SSAs. The null reduction on party 4 is
\begin{equation}
    \ent{124} + \ent{234} \geq \ent{24} + \ent{1234}\,,
\end{equation}
which is itself SSA on subsets $\{124\}$ and $\{234\}$. Finally, the null reduction on party 5 is
\begin{equation}
    \ent{125} + \ent{135} + \ent{135} + \ent{235} \geq \ent{15} + \ent{35} +        \ent{1235} + \ent{1235}\,,
\end{equation}
which is the sum of two SSAs, on subsets $\{125\}$ and $\{135\}$ and on subsets $\{135\}$ and $\{235\}$. Hence every null reduction 
is an HEI despite the fact that the original inequality \eqref{eq:ce-c4} is not.

However, observe that the purification of  \eqref{eq:ce-c4} on party 2
\begin{multline}
    \ent{245} + \ent{245} + \ent{235} + \ent{234} + \ent{135} + \ent{135} + \ent{125} + \ent{124} \geq\\
    \ent{2345} + \ent{13} + \ent{15} + \ent{1245} + \ent{1235} + \ent{35} + \ent{25} + \ent{24} + \ent{24}\,,
\end{multline}
with subsequent null reduction on party 1, gives 
\begin{equation}
    \ent{135} + \ent{135} + \ent{125} + \ent{124} \geq \ent{13}  +\ent{15} + \ent{1245} + \ent{1235}
\end{equation}
which is not an HEI.\footnote{This can be directly checked by evaluating the inequality on $\mathsf{N} = 5$ extreme rays.} 

This observation motivates us to consider a corollary of Theorem \ref{thm:if-HEI-then-contraction} with a slightly weaker converse. 

\begin{cor}
If $(\leftm,\rightm)\in \sbset$ is a {\CPI} then, for all $i,j\in[\N+1]$ ($i\neq j$), the null reduction on $i$ of the purification on $j$ of $(\leftm,\rightm)$ is an HEI.
\end{cor}
\noindent
The converse\footnote{This converse was stated in v1 of this paper as Conjecture 4'.} to this corollary is a weaker statement than Conjecture \ref{conj:4}: we let the purifier vary here, so there are more null reductions that are assumed to be HEIs. Specifically, since switching the purifier with the central party does not change whether a contraction map exists (see appendix \ref{sec:repurification}), there are $\N+1\choose 2$ distinct non-trivial null reductions starting from a given superbalanced inequality (versus $\N$ if the purifier is fixed). Unfortunately, this converse, too, is false. Consider the following non-HEI\footnote{The inequality can be shown not to be an HEI by direct computation against one of the known extreme rays in the $\N = 6$ list of \cite{hecdata}. The violating ray is the permutation $(136)(2475)$ of ray number 481 in the list, which has entropy vector: 
\begin{multline}
(2, 2, 1, 2, 3, 2, 4, 3, 4, 5, 4, 3, 4, 5, 4, 3, 4, 3, 5, 4, 5,
 5, 6, 5, 6, 5, 6, 5, 7, 6, 5, 5, 6, 5, 7,
\\6, 7, 6, 5, 6, 5, 7,
 4, 5, 5, 6, 5, 6, 5, 6, 5, 6, 5, 6,5,6,4,5,4,3,4,4,2).
\end{multline}
Its dot product with the dual vector associated to the above inequality is $-2$.} superbalanced inequality:
\begin{multline}
\ent{34} + \ent{124} + \ent{156} + \ent{356} + \ent{456}
+ \ent{1235} + \ent{1236} + \ent{1245} \\
+ \ent{1346} + \ent{2346} + \ent{2456}
\geq
\ent{1} + \ent{2} + \ent{3} + \ent{4}
+ \ent{15} + \ent{36} + \ent{46} + \ent{56} \\
+ \ent{245} + \ent{1234} + \ent{3456}
+ \ent{12346} + \ent{12356} + \ent{12456},
\end{multline}
all of its null reductions are valid HEIs, as it can be directly checked.

\subsection{Contraction is equivalent to inclusion dominance}
\label{ssec:cond_incdom}

We now want to show that a pair $(\leftm,\rightm) \in \cset$ admits a contraction map if and only if it admits an inclusion dominance map. We will actually prove something stronger, namely that a contraction map is equivalent to an inclusion dominance map.

\begin{thm}\label{thm:contraction-equals-inclusion-dom}
Given $(\leftm,\rightm) \in \cset$, a contraction map for $(\leftm,\rightm)$ is an inclusion dominance map and vice versa.
\end{thm}
\begin{proof}

\textbf{Inclusion dominance is contraction:} Let $f$ be an inclusion dominance map for $(\leftm,\rightm)$. Our claim is that $f$ is a contraction map. So we need to show that the map $f$ has the contraction property
\begin{equation}
     \ham{u}{u'} \geq \ham{f(u)}{f(u')}\,,
\end{equation}
for every pair of bit-strings $u,u'\in\{0,1\}^m$, and that it obeys the constraint
\begin{equation}
    f(\leftocc_{(i)}) = \rightocc_{(i)}.
\end{equation}
for every  $i\in[\N+1]$.

First we show that $f$ obeys the contraction property. Let $u, u' \in \{0,1\}^m$ be two arbitrary bit-strings. Consider the vector $u \wedge u'$ which, by definition, is a subset of both $u$ and $u'$. Therefore, since $f$ is an inclusion dominance map, $f(u \wedge u')$ must be a subset of both $f(u)$ and $f(u')$. This implies that 
\begin{equation}
    |f(u \wedge u')| \leq f(u) \cdot f(u') \, ,
\end{equation}
for otherwise the vector $f(u \wedge u')$ cannot have complete overlap with both $f(u)$ and $f(u')$. By virtue of $f$ being a dominance map, we have
\begin{equation}
    |u \wedge u'|=|f(u \wedge u')|
\end{equation}
 and therefore
\begin{equation}
    u\cdot u' = |u \wedge u'|=|f(u \wedge u')| \leq f(u) \cdot f(u')\,.
\end{equation}
By  \eqref{eq:Haming-dot} and the property $|u| = |f(u)|$ of dominance maps, this is equivalent to\footnote{
    Note that the inequality changes direction in going from the dot-product to the Hamming-distance description, due to the negative sign in \eqref{eq:Haming-dot}.
} 
\begin{equation}
    \ham{u}{u'} \geq \ham{f(u)}{f(u')}\,,
\end{equation}
manifesting contraction as desired.

Next, we show that $f(\leftocc_{(i)}) = \rightocc_{(i)}$ for all $i\in[\N+1]$. Let $u = \leftocc_{(i)}$ with $i \in [\N + 1]$
in \eqref{eq:UVdomi}. Then
\begin{equation*}
    |\leftocc_{(i)}| = \leftocc_{(i)} \cdot \leftocc_{(i)} \leq f(\leftocc_{(i)}) \cdot \rightocc_{(i)} \leq |\rightocc_{(i)}|=|\leftocc_{(i)}|\,.
\end{equation*}
where the first inequality follows from the second property of dominance maps, the second inequality from the fact that these are $(0,1)$ vectors, and the last equality from balance. Thus $f(\leftocc_{(i)}) \cdot \rightocc_{(i)} = |\rightocc_{(i)}|$, which implies $f(\leftocc_{(i)})\geq \rightocc_{(i)}$. Combined with $|f(\leftocc_{(i)})|=|\leftocc_{(i)}|=|\rightocc_{(i)}|$, this proves that $f(\leftocc_{(i)})=\rightocc{(i)}$.

\paragraph{Contraction is inclusion dominance:} 
Let $f$ be a contraction map for $(\leftm,\rightm)$. We claim that $f$ is an inclusion dominance map, in other words that:
\begin{enumerate}[(1)]
    \item $|u| = |f(u)|$ for all $u\in\{0,1\}^m$, and
    \item $u\cdot \leftocc_{(i)} \leq f(u) \cdot \rightocc_{(i)}$ for all $u\in\{0,1\}^m$ and for all $i\in[\N+1]$, and
    \item $\forall u,u'$ with $u' \leq u$, we have  $f(u') \leq f(u)$.\end{enumerate}
Property $(1)$ follows from simultaneous contraction for two extremal bit strings, $0$ and $e$, similarly to the strategy employed in the proof of theorem \cref{thm:if-HEI-then-contraction}. Denoting the purifier by $p$, we write the following two contraction inequalities:

\begin{equation}
 \ham{u}{\leftocc_{(p)}} \ge \ham{f(u)}{f(\leftocc_{(p)})} \,,\qquad \ham{u}{\leftocc_{(1)}} \ge \ham{f(u)}{f(\leftocc_{(1)})} \,.
\end{equation}
The first one implies, by virtue of $\leftocc_{(p)}={0}$ and $f(\leftocc_{(p)}) = \rightocc_{(p)} = {0}$, that
\begin{equation}\label{eq:leq-direction}
       |u|\ge |f(u)| \,,
\end{equation}
while the second one implies, by virtue of $\leftocc_{(1)}=e$ and $f(\leftocc_{(1)}) = \rightocc_{(1)} =e$, that
\begin{equation}\label{eq:geq-direction}
    -|u| + m\ge-|f(u)|+m \quad \Rightarrow\quad |u|\le|f(u)| \,.
\end{equation}
Combining \eqref{eq:leq-direction} with \eqref{eq:geq-direction} we obtain $| u | = | f(u) |$, proving property (1). 

To show property (2), we use the contraction property on $u$ and $\leftocc_{(i)}$,
\be
    |u-\leftocc_{(i)}| = |u| + |\leftocc_{(i)}| - 2 \, u \cdot \leftocc_{(i)}
    \ge
    |f(u) -f(\leftocc_{(i)})| = |f(u)| + |\rightocc_{(i)}| - 2 \, f(u) \cdot \rightocc_{(i)}\,.
\ee
    Using the above-proved property (1) and balance this immediately gives the desired result,
\be
    u\cdot \leftocc_{(i)} \leq f(u) \cdot \rightocc_{(i)} \, .
\ee

Finally, to show property (3), consider arbitrary $u'\leq u$.  Again using property (1) proven above, $|u|=|f(u)|$ and $|u'|=|f(u')|$, which along with contraction property then gives $u\cdot u' \leq f(u) \cdot f(u')$. Together we obtain
\begin{equation}
    |u'|=u' \cdot u \leq f(u) \cdot f(u') \leq |f(u')| = |u'|.
\end{equation}
Thus $f(u) \cdot f(u') = |f(u')|$ which means that $f(u')\leq f(u)$. This ends the proof that every contraction map for a centered inequality is an inclusion dominance map.
\end{proof}

This theorem has two immediate corollaries:
\begin{cor}\label{thm:if-centered-hei-then-dominance}
A contraction map for $(\leftm,\rightm) \in \cset$ is a dominance map.
\end{cor}
\begin{proof}
This follows from Theorem \ref{thm:contraction-equals-inclusion-dom} since every inclusion dominance map is a dominance map.
\end{proof}
\begin{cor}\label{cor:HEI-implies-dominance}
If $(\leftm,\rightm)\in \cset$ is a \CPI, then it obeys dominance.
\end{cor}

\paragraph{Counterexamples to converses:} The converse of Corollary \ref{thm:if-centered-hei-then-dominance} does not hold: not every dominance map for a centered inequality is a contraction map, since, as we saw in Example \ref{eg:non-example-inclusion-dom}, not every dominance map is an inclusion dominance map.

The converse of Corollary \ref{cor:HEI-implies-dominance} also does not hold: a centered inequality obeying dominance need not be a \CPI. A counterexample is given by the inequality from Example \ref{eg:non-example-inclusion-dom}, which is centered and obeys dominance but fails inclusion dominance. By Theorem \ref{thm:contraction-equals-inclusion-dom}, this implies the non-existence of a contraction map.  We can also see more directly that it is not an HEI, being violated by the entropy vector of the following graph (with all edges having unit weight):
\[
\begin{tikzpicture}[scale=0.6,
    every node/.style={circle, draw, inner sep=1pt},
    dot/.style={circle, fill=black, inner sep=2.5pt, draw=none},
    txt/.style={inner sep=0pt, outer sep=1pt, draw=none}
]
    \node (1) at (0,1.5) {1};
    \node (2) at (0,0)   {3};
    \node (3) at (0,-1.5){5};

    \node[dot] (V) at (1.3,0) {};

    \node (4) at (4,1.5) {2};
    \node (5) at (4,0)   {4};
    \node (6) at (4,-1.5){6};

    \node[dot] (W) at (2.7,0) {};

    \draw (1) -- (V);
    \draw (2) -- (V);
    \draw (3) -- (V);

    \draw (4) -- (W);
    \draw (5) -- (W);
    \draw (6) -- (W);

    \draw (V) -- (W);
\end{tikzpicture}
\]

\subsection{Dominance implies PCM}
\label{ssec:dom_PCM}

Here we prove that for centered inequalities, dominance implies positive-combinations majorization (cf.~\cref{def:PCM}).

\begin{thm}\label{thm:if-dominance-then-PCM}
Let $(\leftm,\rightm)\in \cset$. If $\leftm$ is dominated by $\rightm$ then $\leftm \preceq^{\rm PC} \rightm$.
\end{thm}
\begin{proof}
Let us prove for arbitrary $v \in \R^{\N+1}_+$ that $v\leftm \preceq v\rightm$. Recall that $\leftm$ and $\rightm$ have the same dimensions, $(\N+1) \times m$ and 
observe that for any $1 \leq k \leq m$ any sum of $k$ coordinates of $v\leftm$ is given by $v\leftm u^T$ for an appropriate bit string $u \in \{0, 1\}^{m}$ with $|u|=k$. Therefore, the sum of the $k$ largest ones is
\begin{equation}
    \sum_{n=1}^k (v \leftm)^\downarrow_n = \max_{u\in B_k}
    v\leftm u^T\,,
\end{equation}
where
\be
B_k:=\left\{u\in\{0,1\}^m:|u|=k\right\}.
\ee
The same is true for $v\rightm$. Recall that for a dominance map $f$ we have $\leftm u^T\le\rightm f(u)^T$. Moreover, since $v \geq 0$, we also have
\begin{equation}
    v\leftm u^T \leq v\rightm f(u)^T\,.
\end{equation}
Thus we obtain
\begin{equation}
    \sum_{n=1}^k (v \leftm)^\downarrow_n = \max_{u\in B_k}
    v\leftm u^T \leq \max_{u\in B_k}
    v\rightm f(u)^T \leq \max_{u\in B_k}
    v\rightm u^T = \sum_{n=1}^k (v \rightm)^\downarrow_n\,,
\end{equation}
where the second inequality holds because we are taking the maximum over a larger set: 
\begin{equation}
\{f(u) : u \in B_k\}\subseteq B_k\,.
\end{equation}
Finally, balance ensures that $\sum\limits_{n=1}^{m} (v \leftm)_n = \sum\limits_{n=1}^{m} (v \rightm)_n.$
\end{proof}
Note that, in the proof, we did not use the existence of rows of all 0s or all 1s, so the above statement is true for arbitrary balanced $(0,1)$-matrices of the same size.

\paragraph{Counterexample to Conjecture \ref{conj:2}:} The counterexample to the converse of Theorem \ref{thm:if-HEI-then-contraction} given in subsec.\ \ref{sssec:converses} is a superbalanced inequality that is not an HEI, yet all of its null reductions are {\CPI}s. Via Theorems \ref{thm:contraction-equals-inclusion-dom} and \ref{thm:if-dominance-then-PCM}, those null reductions therefore obey PCM, falsifying Conjecture \ref{conj:2}. 

One may wonder whether a weakening of Conjecture \ref{conj:2} might nonetheless be true. Specifically, if all its null reductions for all choices of purifier obey PCM, is the pair $(\leftm,\rightm)\in\sbset$ necessarily an HEI? (This is ${\N+1}\choose{2}$ null reductions, versus $\N$ null reductions with a fixed purifier as in Conjecture \ref{conj:2}.) It turns out the answer is no. Here is a counterexample, an $\N = 6$ superbalanced non-HEI,\footnote{The inequality \eqref{eq:ce-c2} is not an HEI because it is the 35th HEI in the $\N=6$ list in \cite{hecdata} from which we subtracted MMI on parties \{1\}, \{3,5\} and \{4\}.} for which all 21 null reductions obey PCM:
\begin{multline}\label{eq:ce-c2}
    \ent{13} + \ent{156} + \ent{235} + \ent{246} + \ent{356} + \ent{456}
    + \ent{1245} + \ent{1345} \\
    \geq \ent{1} + \ent{2} + \ent{3} + \ent{35} + \ent{46} + \ent{56}
    + \ent{145} + \ent{1356} + \ent{2456} + \ent{12345}\,.
\end{multline}

\paragraph{Counterexample to converse of Theorem \ref{thm:if-dominance-then-PCM}:} The converse of Theorem \ref{thm:if-dominance-then-PCM} does not hold, as can be explicitly checked by the following counterexample. Define the following matrices:
\begin{equation}
    \label{01counterexample}
\mathbf{A} = \begin{pmatrix}
1 & 1 & 0 \\
1 & 1 & 0 \\
1 & 0 & 1 \\
1 & 0 & 1 \\
0 & 1 & 1 \\
0 & 1 & 1
\end{pmatrix},\qquad \mathbf{B} = \begin{pmatrix}
1 & 1 & 1 & 0 & 0 & 1 & 0 & 0 \\
1 & 0 & 0 & 1 & 1 & 1 & 0 & 0 \\
1 & 1 & 0 & 1 & 0 & 0 & 1 & 0 \\
1 & 0 & 1 & 0 & 1 & 0 & 1 & 0 \\
1 & 1 & 0 & 0 & 1 & 0 & 0 & 1 \\
1 & 0 & 1 & 1 & 0 & 0 & 0 & 1
\end{pmatrix}.
\end{equation}
From these two matrices, define $(\leftm,\rightm)$:

\begin{equation}
    \leftm = \begin{pmatrix}
        \mathbf{A}\, & \mathbf{I}\, & \, \mathbf{I}\\
        1&\cdots  &1\\
        0&\cdots &0
    \end{pmatrix},\qquad \rightm = \begin{pmatrix}
        \mathbf{B} & & \mathbf{O}\\
        1& \cdots &1\\ 
        0& \cdots &0
    \end{pmatrix},
\end{equation}
where $\mathbf{I}$ is the $6\times 6$ identity matrix and $\mathbf{O}$ is a $6\times7$ block of 0s.  $(\leftm,\rightm)$ is centered: the matrices have equal numbers of rows (8) and columns (15), equal row sums, and a row of all 1s.

Let us show that $\leftm$ is not dominated by $\rightm$. Consider $U = \{1, 2, 3\}$. Then 
\begin{equation}
    \leftm^{(U)}e^T 
= (2,2,2,2,2,2,3,0)^T\,.
\end{equation}
Observe that there is no $V$ of size $3$ such that $\leftm^{(U)}e^T \leq \rightm^{(V)}e^T$. Indeed, the best we can do is to have $1 \in V$, but then every other pair of columns of $\rightm$ has one position (other than the last row) where both entries are zero, thus failing to dominate the corresponding position of $\leftm^{(U)}e^T$ which is $2$.

On the other hand, $\leftm \preceq^{\rm PC} \rightm$, as we now show. It can be checked directly that the $U$ above is the only set that fails the dominance property. Therefore, in order to show PCM in the spirit of Theorem \ref{thm:if-dominance-then-PCM}, it is sufficient to ensure that for any $v \in \mathbb{R}^8_{+}$ there exists a set $V$ of size $3$ such that $v\leftm^{(U)}e^T \leq v\rightm^{(V)}e^T$. Indeed, we can take $V = \{1,2,3\}$ if $v_1 \geq v_2$, and $V = \{1,4,5\}$ otherwise.

\subsection{PCM implies BCM}\label{ssec:pcm-bcm}

It is clear from the definitions that positive-combinations majorization $\preceq^{\rm PC}$ implies binary-combinations majorization $\preceq^{\rm BC}$, since the conditions required for BCM (\cref{def:BCM}) is a strict subset of the conditions required for PCM (\cref{def:PCM}). 
\paragraph{Counterexample to BCM implying PCM:} The converse is not true even for $(0, 1)$-matrices, as can be explicitly checked by the following counterexample. Define the following matrices:
\begin{equation}
   \mathbf{A}  = \left(\begin{array}{ccc}
 1 & 1 & 0\\
 1 & 1 & 0\\
 1 & 0 & 1\\
 1 & 0 & 1\\
 0 & 1 & 1\\
 0 & 1 & 1
\end{array}\right) ,\qquad \mathbf{B} = \left(\begin{array}{cccccccccccc}
 1 & 1 & 0 & 0 & 1 & 0\\
 1 & 1 & 1 & 0 & 0 & 1\\
 1 & 1 & 1 & 1 & 0 & 0\\
 1 & 0 & 1 & 0 & 1 & 0\\
 1 & 1 & 0 & 1 & 0 & 1\\
 1 & 0 & 0 & 1 & 0 & 1
\end{array}\right).
\end{equation}
From these two matrices, define $(\leftm,\rightm)$:
\begin{equation}\label{eq:BCMnotPCM}
    \leftm = \begin{pmatrix}
        \mathbf{A}\, & \mathbf{I}\, & \, a_2 & \, a_3 & \, a_5\\
        1&\cdots & \cdots & \cdots &1\\
        0&\cdots & \cdots & \cdots &0
    \end{pmatrix},\qquad \rightm = \begin{pmatrix}
        \mathbf{B} & & \mathbf{O}\\
        1& \cdots &1\\ 
        0& \cdots &0
    \end{pmatrix},
\end{equation}
where $\mathbf{I}$ is the $6\times 6$ identity matrix, $\mathbf{O}$ is a $6\times 6$ block of 0s, and $a_i$ is a column of size $6$ with $1$ in row $i$ and $0$s everywhere else. $(\leftm,\rightm)$ is centered: the matrices have equal numbers of rows (8) and columns (12), equal row sums, and a row of all 1s.

It can be checked by enumeration of all $(0,1)$ vectors that $\leftm \preceq^{\rm BC} \rightm$. On the other hand, $\leftm \not\preceq^{\rm PC} \rightm$, since for $v = \begin{pmatrix}
    2 & 1 & 1 & 3 & 1 & 3 & 0 & 0
\end{pmatrix}$ \begin{equation}\label{eq:BCMnoPRCM:vLvR}
    v\leftm = \left(\begin{array}{cccccccccccc}
7 & 7 & 8 & 2 & 1 & 1 & 3 & 1 & 3 & 1 & 1 & 1
\end{array}\right) \not\preceq \left(\begin{array}{cccccccccccc}
11 & 5 & 5 & 5 & 5 & 5 & 0 & 0 & 0 & 0 & 0 & 0
\end{array}\right)=v\rightm\,,
\end{equation}
which can be observed by looking at the sum of the $3$ largest entries.

\subsection{BCM implies region dominance}\label{ssec:bcm-rd}

Here we prove that for centered inequalities, binary-combinations majorization implies region dominance.

\begin{thm}\label{thm:if-bcm-then-rd}
Let $(\leftm,\rightm) \in \cset$. If $\leftm  \preceq^{\rm BC} \rightm$ then $\leftm$ is region-dominated by $\rightm$.
\end{thm}
\begin{proof}
We will in fact show that this statement holds for arbitrary $(0, 1)$-matrices $\leftm, \rightm$ of the same size; in other words, we will not use the existence of a row of all 0s or a row of all 1s. If $\leftm \preceq^{\rm BC} \rightm$, then considering every unit vector for the role of $v$ in in Definition \ref{def:BCM}, we obtain that $(\leftm,\rightm)$ is balanced; in other words, balance on party $i$ follows from the majorization condition with $v_j=\delta_{ij}$.

Now let us prove region dominance \eqref{eq:region-dominance} for arbitrary $X \subseteq [\N+1]$. Namely, we prove that the number of columns of $\leftm$ whose support contains $X$ doesn’t exceed that of $\rightm$. Consider $k = \Big| \bigwedge_{i \in X} \leftocc_{(i)} \Big|$, i.e.\ the number of terms on the LHS containing $X$, in \eqref{eq:majdef}. If $k = 0$, then there is nothing to prove. Let $x$ be the index bit-string for $X$, that is, 
$x_i=1$ if $i\in X$ and 0 otherwise. Notice that $x\leftm^{(n)} \leq |X|$ for any column $\leftm^{(n)}$ and the equality is achieved if and only if the support of $\leftm^{(n)}$ contains $X$. The same is true for $\rightm$. Then, by Definition  \ref{def:BCM}, we obtain
\begin{equation}
k|X|=\sum_{n=1}^k\limits(x\leftm)^\downarrow_n
\leq\sum_{n=1}^k\limits (x\rightm)^\downarrow_n\,,
\end{equation}
implying that the number of columns in $\rightm$ whose support contains $X$ is at least $k$.
\end{proof}

\paragraph{Counterexamples to converses:} The converse to theorem \ref{thm:if-bcm-then-rd} also does not hold, as shown by the following counterexample:
    \begin{equation}
    \label{regiondom}
        \leftm = \begin{pmatrix}
1 & 1 & 1 & 1  \\
1 & 1 & 0 & 0 \\
1 & 0 & 1 & 0 \\
0 & 1 & 1 & 0 \\
0 & 0 & 0 & 0
\end{pmatrix},\quad \rightm = \begin{pmatrix}
    1 & 1 & 1 & 1\\
    1 & 1 & 0 & 0 \\
1 & 0 & 1 & 0 \\
1 & 0 & 0 & 1 \\
0 & 0 & 0 & 0
\end{pmatrix}.
    \end{equation}
Then $(\leftm, \rightm)\in \cset$. Observe that $\leftm \not\preceq^{\rm BC} \rightm$ since, for $v = \begin{pmatrix}
        0 & 1 & 1 & 1 & 0
    \end{pmatrix}$,
\begin{equation}
    v\leftm = \begin{pmatrix}
        2 & 2 & 2 & 0
    \end{pmatrix} \not\preceq \begin{pmatrix}
        3 & 1 & 1 & 1
    \end{pmatrix}= v \rightm\,.
\end{equation}
However, $(\leftm, \rightm)$ satisfies region dominance.

Under the operation of changing the purifier to be the central region, or equivalently taking the complement of the matrices $\leftm,\rightm$, the BCM property is preserved but region dominance is not. (This operation is discussed more in appendix \ref{sec:repurification}.) Indeed, applying this operation to the previous example yields an inequality that obeys BCM but not region dominance, as the reader can readily verify. This raises the question of whether a slightly weaker statement than the above converse might hold, namely: If a centered inequality and its repurification on the central party both obey region dominance, then BCM holds. The following counterexample shows that this too is false:
\begin{equation}
\leftm=
\begin{pmatrix}
1&1&1&1&1&1\\
1&0&0&1&1&1\\
0&1&0&1&1&0\\
0&0&1&1&0&1\\
0&0&0&0&1&1\\
0&0&0&0&0&0
\end{pmatrix},
\quad
\rightm=
\begin{pmatrix}
1&1&1&1&1&1\\
0&1&1&0&1&1\\
0&1&0&1&0&1\\
0&0&1&1&0&1\\
0&0&0&0&1&1\\
0&0&0&0&0&0
\end{pmatrix}.
\end{equation}
Both $(\leftm, \rightm)$ and its repurification on party 1 satisfy region dominance, but not BCM: for $v = (0,1,1,1,1,0)$ we obtain
\be
v\leftm = (1,1,1,3,3,3)\not\preceq (0,2,2,2,2,4)=v\rightm\,.
\ee

\subsection{Balanced HEIs obey region dominance}\label{ssec:bhei-rd}

Although our primary focus has been on centered inequalities, here we prove more directly that every balanced HEI (not necessarily centered) is region dominant. A notable feature of this proof is that, unlike in the rest of the paper, we directly use the definition of an HEI as an inequality that holds for the min-cut function on an arbitrary graph, rather than relying on the existence of a contraction map.

\begin{figure}
    \centering
    \includegraphics[width=0.85\linewidth]{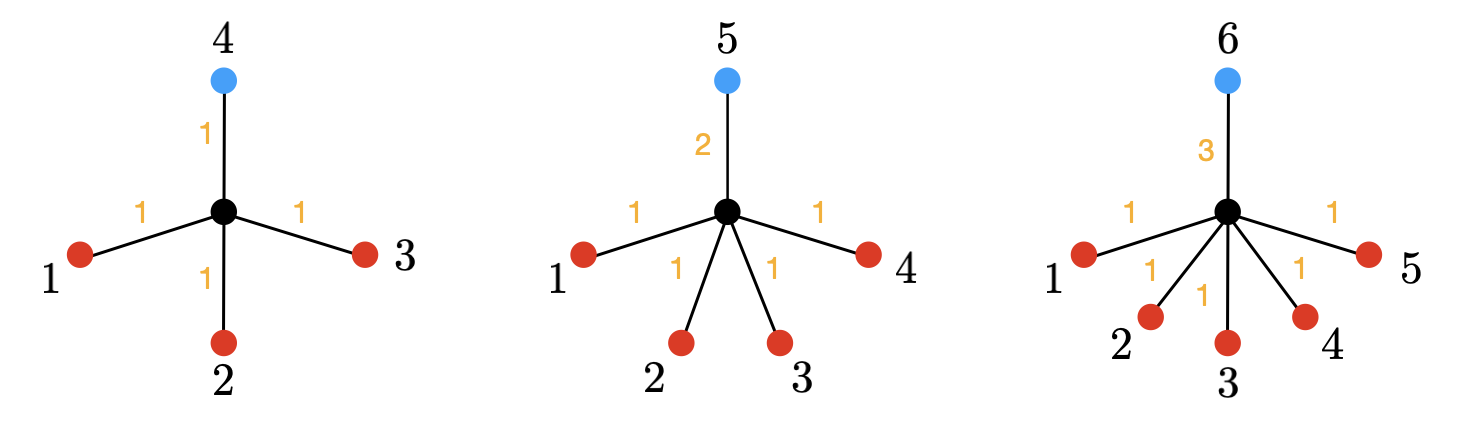}
    \caption{Examples of star graphs for $k = 3,4,5$ respectively, used in the proof of \cref{thm:if-balanced-hei-then-region-dominance}, with the edge weights shown in orange.  When $k\neq \mathsf{N}$, all the remaining vertices are isolated and disconnected (and hence omitted from the diagram for visual clarity).}
    \label{fig:stargraphs}
\end{figure}

\begin{thm}\label{thm:if-balanced-hei-then-region-dominance}
Let $(\leftm,\rightm) \in \bset$ be an HEI. Then $\leftm$ is region-dominated by $\rightm$.
\end{thm}
\begin{proof}
We seek to prove that $(\leftm,\rightm)$ is region dominant for any $k$-subset of $[\N+1]$, for any $2 \leq k \leq \mathsf{N}$. (The $k = 1$ case is ensured by the assumption that the HEI is balanced.) 

We begin by considering $k \geq 3$; the case for $k = 2$ will be taken care of at the end. Recall the definition of region dominance in \eqref{eq:region-dominance}, where $X$ in the definition is an arbitrary collection of parties. Without loss of generality, permute the parties so that $X \coloneq \{123\cdots k\}$ along with the appropriate permutation of $(\leftm, \rightm)$.  Consider a star graph on the set of vertices $V = \{1,2,3,\dots,k,\N+1\}$, where we pick the edge associated to the vertex $\N+1$ to be the heavier one with weight $k-2$. (See fig.\ \ref{fig:stargraphs} for some examples.)
    
In the following, $\alpha_\ell$ is any $\ell$-subset (including the 0-subset $\emptyset$) of $\{1,2,3,\dots, k\}$ and $\{\dots\}$ denotes any subset of $\{k+1,\dots,\N\}$. For any given subsystem, $X_\ell = \alpha_\ell \cup \{\dots\}$, the min-cut function on such graph takes the following values:
\begin{equation}\label{eq:min-cut-star-graph}
    \ent{X_\ell} = \min(\ell, 2k - \ell - 2) \, .
\end{equation}
Then, evaluating the inequality on such graph we obtain
\begin{equation}
    \sum_{\alpha_1} x_{\alpha_1} + 2\left(\sum_{\alpha_2} x_{\alpha_2}\right) + 3\left(\sum_{\alpha_3} x_{\alpha_3}\right) + \cdots + (k-2) x_{\{123\cdots k\}} \geq (x\leftrightarrow y)
\end{equation}
where $x_{\alpha}$ is the number of LHS terms with form $\alpha \cup \{\dots\}$, and similarly for $y_\alpha$ but for the RHS. (Note that the form of \eqref{eq:min-cut-star-graph} implies that for $\ell \geq k-1$ the value of the entropy starts to decrease linearly, which is why the last coefficient in the above is $k-2$ compared to just $k$.) However, balance implies that, for any $\ell \in [k]$,
\begin{equation}
    x_{\{\ell\}} +  \sum_{\alpha_2 \ni \ell} x_{\alpha_2} + \sum_{\alpha_3 \ni \ell} x_{\alpha_3} + \cdots + x_{\{123\cdots k\}} = y_{\{\ell\}} +  \sum_{\alpha_2 \ni \ell} y_{\alpha_2} + \sum_{\alpha_3 \ni \ell} y_{\alpha_3} + \cdots + y_{\{123\cdots k\}}\,.
\end{equation}
Summing them, we obtain
\begin{equation}
    \sum_{\alpha_1} x_{\alpha_1} + 2\left(\sum_{\alpha_2} x_{\alpha_2}\right) + 3\left(\sum_{\alpha_3} x_{\alpha_3}\right) + \cdots + k \,x_{\{123\cdots k\}} = (x \leftrightarrow y)\,.
\end{equation}
And subtracting it from the inequality above we get the desired property
\begin{equation}
    x_{\{123\cdots k\}} \leq y_{\{123\cdots k\}}\,.
\end{equation}
This proves the statement for $k\geq 3$.

For $k = 2$, we also consider the star graph on vertices $V = \{1,2,3\}$ with all edges unit weight. Otherwise, the proof is identical to the above. 
\end{proof}

\paragraph{Counterexample to the converse:} It's easy to see that the converse direction does not hold. A counterexample is provided by Example \ref{eg:eg-region-dominance}. The example obeys region dominance (as can be explicitly checked) but it is not a balanced HEI. Indeed, the inequality of Example \ref{eg:eg-region-dominance} can be obtained by taking the following primitive $\N = 6$ HEI,\footnote{This HEI is the 145th in the $\N=6$ list in \cite{hecdata}.}
\begin{multline}
    \ent{123} + \ent{124} + 
    \ent{134}+
    \ent{136} + \ent{145} +
    \ent{156} +
    \ent{236} + \ent{245} +
    \ent{3456}
    \geq\\
    \ent{1} +
    \ent{2} + \ent{13} + \ent{14} + \ent{36} + \ent{45} +
    \ent{56} +
    \ent{234} + \ent{1236} + \ent{1245}+
    \ent{13456}\,,
\end{multline}
and subtracting from it the HEI
\begin{equation}
    \ent{134} + \ent{3456} + \ent{156} \geq \ent{1} + \ent{34} + \ent{56} + \ent{13456}\,,
\end{equation}
which is MMI involving \{1\}, \{3,4\}, and \{5,6\}. The resulting inequality is the one given in Example \ref{eg:eg-region-dominance}. Being the difference of two primitive HEIs, it cannot be an HEI.

\section{Discussion}\label{sec:discussion}

\subsection{Summary}
\label{ssec:discussion-summary}

In this work we have examined combinatorial properties of holographic entropy inequalities, focusing primarily on the centered ones, which can be obtained by null reduction of superbalanced HEIs.  Such inequalities are balanced and have the same number of terms on both sides of the inequality.  Curiously, they admit a number of rather intriguing properties with nested relational structure between them, summarized in \cref{fig:logic-map}.
We have seen that non-negativity of the corresponding information quantity is equivalent to inclusion dominance, which through a chain of increasingly weak properties (dominance, PCM, and BCM) implies region dominance. In a sense, these properties restrict and quantify the building blocks of such HEIs.  These entail composing individual parties to build up larger subsystems on the one hand, and composing individual terms (each corresponding to a specific subsystem) to build up the HEI on the other.  Expressed in terms of our $(0,1)$ matrices $(\leftm,\rightm)$, the former combines the rows (each row describing a single party), whereas the latter combines the columns (each column describing a single term).  The next two paragraphs discuss the corresponding properties from a more intuitive standpoint.

\paragraph{Dominance properties:}
At the bottom of the implication ladder, namely \emph{region dominance}, one takes into account all the terms in the HEI, but only a given subset of parties. More precisely, for region dominance one considers a specific multi-party subsystem (i.e.\ a ``region'') and compares the number of terms containing it on the two sides of the inequality.  In the matrix language, for a fixed subset of rows of $\leftm$ and $\rightm$, one compares the number of columns which have 1 in each of these rows. For \emph{dominance}, on the other hand, one considers the total occurrence of individual parties in a subset of terms; picking a subset of columns of $\leftm$, one seeks an equal-sized subset of columns of $\rightm$ having at least as many occurrences of each party, i.e.\ at least as many 1's in each row. Dominance however doesn't care about the distribution of the parties across terms, which in turn comes at the top of the ladder, namely \emph{inclusion dominance}.  This condition requires the parties to be distributed  in such a way as to produce a nested structure of dominance, suggestive, in a sense, of ``building up'' the HEI.  

\paragraph{Majorization properties:}
Let us now turn to the matrix majorization properties necessary for dominance and sufficient for region dominance.  Historically, the entry point \cite{Grimaldi:2025jad} into the combinatorial properties came through majorization relation between two positive vectors (pertaining to configurations with all regions localized on a light cone; see below for further discussion of the implications of these light-cone configurations).  This is equivalent to \emph{positive-combinations majorization} (PCM), and in fact constitutes the only non-purely-combinatorial criterion formulated here. PCM is the property indicated in  Conjecture \ref{conj:1}, ensured by null-reducing any HEI.  It might therefore seem a pity that we have disproved the converse statement, Conjecture \ref{conj:2}, since this quenches the original hope in \cite{Grimaldi:2025jad}  that understanding the light-cone configurations for all choices of central regions would suffice to ``bootstrap'' to fully general (spacelike) configurations.  However, one may seek to turn this into a virtue, by characterizing the reason for the failure of Conjecture \ref{conj:2} and using it to formulate the missing piece for the putative bootstrap.

The weaker majorization property, namely \emph{binary-combinations majorization} (BCM), was initially formulated simply as a mathematically interesting extension, which at first sight might seem to have little physical significance.  However, it forms the glue, or transition, between the column behavior and the row behavior.  This transition does not occur through individual matrix elements which are perhaps too ``atomistic'', but rather through combining subsets of both collections.  More specifically, for any given subsystem $Z$, we construct a vector of sizes of overlaps $\abs{Z \cap X_n}$ with the individual terms $X_n$ on the LHS and require it to be majorized by the corresponding vector $\abs{Z \cap Y_n}$ for the RHS terms.  In other words, here we are retaining some information regarding not just the individual parties (through choice of $Z$), but also regarding the individual terms (through the components $n$).  Region dominance is then merely a restriction of this statement to the maximal overlap size $\abs{Z}$: majorization dictates that the RHS must admit at least as many such terms as the LHS.

\subsection{Future directions}

\subsubsection{Centered inequalities}

We have seen that the class of centered inequalities possesses many interesting features. A number of questions specifically pertaining to such inequalities are raised by our work.

\paragraph{Nesting:}
We have seen that inclusion dominance can be thought of as dominance supplemented by a certain nesting (or inclusion) structure.  Given the above organizational viewpoint,  one might wonder whether one could analogously formulate a further natural extension of our implication chain, incorporating a nesting structure for rows instead of columns.  
However, unlike the cases of dominance and inclusion dominance, where we specify a LHS subset and try to find a RHS subset satisfying the corresponding criteria, in case of region dominance there are no choices to make, because the meaning of the rows is common between the two sides.  
Nevertheless, since this form of inclusion is reminiscent of  entanglement wedge nesting,  it seems a natural place to seek useful connections to properties of contraction maps.

For example, in recent work \cite{Czech:2026tgj}, it has been suggested that such combinatorial inclusion-like properties of HEIs are the principles that protect the holographic renormalization group. Claim (1.5) of \cite{Czech:2026tgj} follows as a corollary of our Theorem \ref{thm:if-balanced-hei-then-region-dominance}: since balanced HEIs obeys region dominance, every region $X_i$ on the LHS has a region $Y_j$ on the RHS that contains it.
There are many further directions to explore, both at the structural level, as well as at the level of physical implications.

\paragraph{Verification of PCM and BCM:} While verification of BCM can be performed via a direct enumeration of all $(0,1)$-combinations, it is not clear from the definition how to verify PCM. One possible approach can be found in \cite{JoeVerducci}. When restricted to $(0,1)$-matrices, otherwise complex majorization relations often reduce to elegant and explicit combinatorial conditions. For example, linear-combinations majorization (LCM) for $(0,1)$-matrices is equivalent to matrices being equal up to a permutation of columns \cite{DahlGutermanS}. PCM and BCM are clearly more complex relations even for $(0,1)$-matrices.\footnote{The method used to prove that BCM implies region dominance in Theorem \ref{thm:if-bcm-then-rd},  in fact, leads to a complete characterization of LCM for $(0,1)$-matrices \cite{DahlGutermanS}.} These questions will be addressed in future work, where we will present a new general algorithm for verifying PCM, as well as explore conditions specific for PCM and BCM of $(0,1)$-matrices, including more detailed connections to the dominance properties.

\paragraph{Hierarchy of cones for centered HEIs:}
Fixing a central region $i \in [\N+1]$ one can define in entropy space the $CH_i$: the cone of entropy vectors compatible with all HEIs centered on party $i$. It is clear that the $H\subset CH_i$, since all centered HEIs are redundant. In the dual entropy space one can also define the $CH^{*}_i$, the cone of all HEIs centered on party $i$. In this dual space, $CH^{*}_i \subset H^{*}$. Theorem \ref{thm:contraction-equals-inclusion-dom} gives a strong combinatorial constraint on the form of these inequalities through inclusion dominance. It would be interesting and useful to understand the extremal structure of the $CH_i$ (i.e.\ its extreme rays and facets).
While $CH_i$ is not permutation-invariant, we can restore permutation invariance by taking the union of all $CH_i$ over $i\in {[\N+1]}$; it would be likewise interesting to examine the resulting structure in more detail. 

More generally, each combinatorial property studied in this paper defines a set of linear inequalities on vectors in $\mathbf{R}^{2^\N-1}$, and therefore a cone in entropy space. Due to the implications among these properties, these cones are nested. It would be very interesting to understand this structure, as well as the classes of quantum states whose entropy vectors lie inside these different cones.

\subsubsection{Holographic entropy cone}

Notice that the three dominance conditions (inclusion dominance, dominance, and region dominance) were defined for all $(\leftm,\rightm)\in\pairset$, so any information quantity on $\N$-party system can be categorized according to these properties.  Moreover, while the two majorization conditions (PCM and BCM) are formulated for equal-sized matrices, we can extend this to any $(\leftm,\rightm)\in\pairset$ simply by padding the smaller matrix with 0's.  However, it bears emphasizing that the sequence of implications summarized in \cref{fig:logic-map} pertains to the centered inequalities, $(\leftm,\rightm) \in \cset$ (except for balance necessitating region dominance, where \cref{thm:if-balanced-hei-then-region-dominance} is formulated for the more general class $(\leftm,\rightm) \in \bset$). The natural next step is to examine the relations among the combinatorial criteria in the broader  $\sbset$ context, in order to elucidate the HEC itself.

\paragraph{Tripartite form:}
In search of a more compact and revealing repackaging of HEIs, 
\cite{Hernandez-Cuenca:2023iqh} formulated the tripartite form (TF) of an information quantity as one consisting of a positive sum of negative tripartite and conditional tripartite informations, which the authors then used to show certain useful properties of HEIs, and to generate hundreds of new HEIs for $\N=6$.
In \cite{Grimaldi:2025jad}, it was conjectured explicitly that \emph{all} superbalanced HEIs can in fact be expressed in tripartite form, with strong supporting evidence appearing in \cite{Hubeny:2026vlk}.  This would provide a convenient tool for further analysis of HEIs, since it makes many useful results manifest.
One may then ask whether the necessary and/or sufficient conditions for a superbalanced inequality to be an HEI that we have discovered in this paper shed further light on this conjecture.

A natural starting point is to ascertain which criteria are satisfied by a single building block of TF, namely  negative conditional $I_3$. (We exclude the unconditioned $-I_3$ terms, since only when the conditioned-on party is non-empty do we obtain a centered HEI.) More specifically, consider a single term in a TF, say $-I_3(1:2:3\, |\,4)$, at $\N=4$.  
    In the matrix representation, the corresponding (false) inequality is given by
    \begin{equation}
        \leftm = \begin{pmatrix}
            1 & 1 & 1 & 1 \\
            0 & 1 & 1 & 0 \\
            0 & 1 & 0 & 1 \\
            0 & 0 & 1 & 1 \\
            0 & 0 & 0 & 0
            \end{pmatrix},\quad 
        \rightm = \begin{pmatrix}
            1 & 1 & 1 & 1 \\
            1 & 0 & 0 & 1 \\
            0 & 1 & 0 & 1 \\
            0 & 0 & 1 & 1\\
            0 & 0 & 0 & 0
    \end{pmatrix}.
        \end{equation}
Up to a permutation, this is actually an example we've already seen, namely Ex.\ \ref{eg:eg-region-dominance}. Therefore, we know that it fails BCM but obeys region dominance.

Since the conditions are additive, if a single term satisfies a given condition, the full TF expression would likewise satisfy it. However an important caveat here is that this applies to the context of centered HEIs, and for a TF expression to be centered, we need to condition each term on the same party --- which is then \emph{not} an HEI.  Correspondingly, null-reducing superbalanced HEIs does not maintain TF-compatibility.

\paragraph{Building up higher HEIs:}
Over the past few years, we have seen intriguing hints at structural relations across different $\N$, both at the level of primitive HEIs (i.e., HEC facets) \cite{Hernandez-Cuenca:2023iqh,Czech:2022fzb}, as well as at the at the level of HEC extreme rays \cite{Hernandez-Cuenca:2022pst,Hubeny:2025bjo}. In both representations, the expressions at smaller $\N$ appear as building blocks for the larger $\N$ ones, though these empirical observations have not been turned into a successful algorithm for deterministically generating the larger $\N$ HEIs (as opposed to merely providing potential candidate HEIs to be subsequently tested).
It is therefore natural to examine this structure through the lens of our combinatorial properties.
In the matrix language, we augment  the number of parties $\N$ by adding rows.  We already mentioned the trivial ``lifts'', corresponding to transformation 4 (duplicating rows) in \cref{sec:transformations}.  To do something non-trivial, namely create a new row which is not a copy of one of the existing ones, we must satisfy the column requirements, which entails modifying columns (and/or adding new ones)  correspondingly.  For example, if one appends a new party to a LHS term $X$, one needs to choose one term $Y\supset X$ on the RHS and append the new party to that term as well (otherwise region dominance would fail).  It would be interesting to see how restrictive the full collection of our criteria is for building up higher-$\N$ HEIs, and to actually implement it for generating new HEIs systematically.

\paragraph{Extreme rays and holographic graph models:}
To describe the min-cut cone $H_\N$ for arbitrary $\N$, it may in fact be easier to use the 
$V$-representation in terms of extreme rays rather than the $H$-representation in terms of the primitive HEIs.  This is due to the result \cite{Hernandez-Cuenca:2022pst}  that (subject to a certain conjecture\footnote{
    In particular, a strong form of the conjecture \cite[Conj.C3]{Hernandez-Cuenca:2022pst} is that the HEC extreme rays are representable by holographic graph models with tree topology.  This has been shown to be the case for all $\N=5$ ERs and vast majority of the known $\N=6$ ERs, and hitherto not falsified for the remaining ones.  (See also \cite{Hubeny:2025bjo} for further discussion.)
}) the extreme rays can be extracted from a more primal construct, namely the subadditivity cone (SAC), delimited by all instances of subadditivity \eqref{eq:SA}.  Its faces specify the ``pattern of marginal independence'' (PMI), characterizing the collection of decorrelated pairs of subsystems, which thereby distills the essence of the corresponding entanglement structure in a suggestively combinatorial language.  Recently, \cite{Hubeny:2024fjn} has developed a more compact reformulation of the PMI in terms of the so-called correlation hypergraph, which provides a new toolkit for exploring $H_\N$.  For example, given an entropy vector satisfying a certain condition,\footnote{
    Specifically, if the line graph of corresponding correlation hypergraph is chordal, the entropy vector can be realized by a simple tree graph model.  In particular, \cite{Hubeny:2025hst} has recently proved that the algorithm presented in \cite{Hubeny:2025bjo} always succeeds.
} we now have an efficient algorithm \cite{Hubeny:2025bjo}  to obtain the corresponding holographic graph model with tree topology.
Since the data describing the correlation hypergraph is captured by a bit string $\{0,1\}^{\D}$ of length $\D=2^\N-1$, and the graph model ``min-cut'' structure \cite{Hernandez-Cuenca:2022pst} can be expressed in terms of $(0,1)$ matrices, it is natural to ask if these can be characterized in terms of combinatorial properties analogous to those explored here.

For a general polyhedral cone, the passage between $H$-representation and $V$-representation is computationally hard \cite{Khachiyan}, but in the case of $H_\N$ which has physically natural interpretation, the rigidity of its structure might conceivably reduce the complexity significantly.\footnote{
In fact, this might be even more so in case of the SAC, which has the added advantage that we know its $H$-representation fully at all $\N$ (even though finding the $V$-representation directly is quite non-trivial \cite{He:2024xzq})  but since ``most'' of it is not holographically (or even quantumly) realizable, it is not entirely clear if this will ultimately provide further simplification beyond $H_\N$.}  One might then hope to make a more direct contact with the combinatorial properties of the HEIs.

Finally, as exemplified by the proof of \cref{thm:if-balanced-hei-then-region-dominance}, holographic graph models can provide a useful tool for identifying HEI candidates which are in fact invalid.  One may then ask if the present techniques allow one to generalize this construction to delimit the HEC more directly.

\paragraph{Does HRT obey RT HEIs?} The null reduction and  majorization test arose in an investigation of whether the (covariant) HRT formula obeys HEIs, specifically in the setting of so-called light-cone configurations of spatial regions in near-vacuum states \cite{Grimaldi:2025jad}. We now know, via Theorem \ref{thm:if-HEI-then-contraction}, Corollary \ref{cor:HEI-implies-dominance}, and Theorem \ref{thm:if-dominance-then-PCM}, that Conjecture \ref{conj:1} is true: all null reductions of all HEIs obey the majorization test. This provides strong evidence that indeed the HRT formula does obey all HEIs. Furthermore, the proof of Conjecture \ref{conj:1} reveals that HEIs posses an interesting universal combinatorial structure. This combinatorial structure could potentially lead to new insights on the broader question of whether the HRT formula obeys HEIs in general, beyond light-cone configurations.

\paragraph{Physical meaning of the HEC:} Returning to the original holographic context, the HEIs are constraints on the multipartite entanglement structure of semiclassical states in quantum gravity. Interpreting those constraints in a quantum information sense, for example by giving an ansatz for the state that automatically obeys them but is flexible enough to accommodate general semiclassical states, remains an open problem. For the simplest superbalanced HEI, namely MMI \eqref{eq:MMI}, such an ansatz was given in terms of perfect tensor states and called ``bipartite dominance'' \cite{Cui:2018dyq}. (Arguments have been advanced against this conjecture \cite{Akers:2019gcv}, but the issue remains open. See for example the recent works \cite{Mori:2024gwe, Balasubramanian:2025hxg, Naskar:2026zka}.) Other previous work addressing the physical interpretation of the HEIs includes \cite{Hernandez-Cuenca:2023iqh,Czech:2023xed,Czech:2025tds,Czech:2025jnw,Czech:2026tgj,Bao:2026lgo}. One can hope that the new combinatorial perspective on HEIs offered in this work could lead to progress on this crucial issue.

\acknowledgments{
We would like to thank Ning Bao, Bartek Czech, Christian Ferko, Sergio Hern\'andez-Cuenca, and Max Rota for useful conversations and comments on a draft of this paper. 
G.G. and M.H. were supported by the U.S. Department of Energy through award DE-SC0009986. V.H. was supported in part by the U.S. Department of Energy through award DE-SC0009999 and by funds from the University of California. 
G.G. would like to thank the Institut des Hautes Etudes Scientifiques for hospitality while this work was undertaken. ChatGPT Codex (Open AI) was used for searches for counterexamples and other numerical explorations for this paper. Otherwise AI tools were not used in the research or writing of this paper.
}

\appendix

\section{More on contraction maps}
\label{sec:contraction}

\newcommand{\cms}{{\color{violet}{{\text{\small s}}}}}
\newcommand{\cmn}{{\color{violet}{{\text{\small n}}}}}
\newcommand{\cmLs}{{\color{purple}{{\text{\small Ls}}}}}
\newcommand{\cmLn}{{\color{purple}{{\text{\small Ln}}}}}
\newcommand{\cmRs}{{\color{purple}{{\text{\small Rs}}}}}
\newcommand{\cmRn}{{\color{purple}{{\text{\small Rn}}}}}
\newcommand{\cmr}{{\color{purple}{{\text{\small r}}}}}
\newcommand{\cmf}{{\color{purple}{{\text{\small f}}}}}

In this appendix we review several interesting aspects of contraction maps.
In particular, in  \ref{ss:classif}  we highlight the different types of contraction maps which one can consider, depending on the specification of the domain and codomain.
In \ref{ss:primred} we comment on whether and how such maps for primitive HEIs combine to prove redundant ones.

\subsection{Classification of contraction maps}
\label{ss:classif}

\paragraph{Relevant subset of bit strings:}
We start with the most obvious distinction one can place on specifying the domain of the contraction map.  Suppose that our candidate inequality has $\ell$ terms on the LHS.  (Since in the next paragraph we will discuss an alternate option to $\ell=\mL$ used in the main text, we use more generic notation here.)  A priori, there are $2^{\ell}$ possible bit strings of length $\ell$.
In the holographic context, the bit strings can be viewed as a specification of bulk regions: for each RT surface corresponding to one of the LHS subsystems, a given component of the bit string indicates whether the region in question lies inside (1) or outside (0) the corresponding RT surface.  This immediately implies that many bit strings are not realized by any region.
For example, suppose the LHS contains nested subsystems such as $1$ and $12$ and consider a bit string  $x$ for which the region lies inside RT surface for subsystem $1$ but simultaneously outside one for subsystem $12$.  Since by `entanglement wedge nesting' the RT surface (or more precisely the homology region) corresponding to subsystem $1$ must lie inside that for subsystem $12$,  such a bit string cannot be realized; it encodes an empty set.  
Similarly, a bit string indicating a region which is inside the homology region of subsystem $1$ and simultaneously inside the homology region of a disjoint subsystem $2$, cannot be realized.\footnote{
    This follows from the previous observation (assuming complementary recovery): since $1$ is nested inside the complement of $2$, the region inside $1$ must be inside the complement of $2$ and therefore outside $2$. 
    (See also \cite[Thm.1c]{Avis:2021xnz}.)
}

Combining these conditions, we see that from the full set of $2^{2^\N-1}$ bit strings, most\footnote{%
\label{fn:realizablebitstringfrac}
    The number of realizable bit strings for the full collection of all possible subsystems of  an $\N$-party system is tabulated in \cite[eq.(2.13)]{Avis:2021xnz}.  Although it grows rapidly with $\N$:   $2, 4, 12, 81, 2646, 1422564, 229809982112, \ldots$, this growth is far slower than the set of all possible bit strings, which would be 
    $2, 8, 128, 32768, 2147483648,\ldots$.
}
are not realizable.  However, not all subsystems can be invoked on the LHS.\footnote{
    For any HEI in tripartite form, the RHS has at least as many terms as the LHS see also \cite{Hubeny:2026vlk}), which weakly bounds the number of LHS terms by $\ell<2^{\N-1}$.
}  One might hope that, as is the case for e.g.\ MMI, the LHS subsystems have a non-constraining crossing structure, in the sense that the LHS would not admit a nested or a disjoint pair of subsystems, in which case all possible bit strings would correspond to some potentially non-empty region.  However, this is not the case. Most HEIs in fact do have the property that not all LHS bit strings are realizable, despite the contraction map invoking them.\footnote{
    Out of the 1877 primitive HEI orbits at $\N\le6$ discussed in \cite{Hernandez-Cuenca:2023iqh}, there are 1530 containing nested LHS subsystems, 1438 containing disjoint LHS subsystems, and 1156 containing both. In fact, across all individual HEIs, only 2.3\% have all LHS bit strings realizable.
   
} 
This means that many LHS bit strings in the full $\{0,1\}^{\ell}$ domain are unrealizable, and therefore should be irrelevant for the contraction map.  Denoting the full domain for ${\ell}$-term LHS bit strings by $\mathcal{D} = \{0,1\}^{\ell}$ and the  realizable subset by $\mathcal{A} \subset \mathcal{D}$,  we therefore have a non-trivial distinction between two possibilities:  
\begin{itemize}[nosep] 
\item (\cmr): \emph{realizable} bit strings only, $x \in \mathcal{A} $
\item (\cmf): \emph{full} set of bit strings, $x \in \mathcal{D} $
\end{itemize} 
Evidently, if a contraction map exists for case (\cmf), then it necessarily exists for case (\cmr) simply by ignoring the unrealizable bit strings.  On the other hand, the converse does not a priori need to hold. Nevertheless, it was shown in \cite{Avis:2021xnz} that if a contraction map exists for (\cmr) then it is a valid HEI, which by the completeness assumption would imply a contraction map likewise exists for (\cmf) (though not necessarily as a simple extension from $\mathcal{A} \to \mathcal{D}$). It is an interesting question whether a realizable contraction map can be directly extended to a full one; we return to this issue in upcoming work. Restricting the domain from (\cmf) to (\cmr) generally constitutes an enormous reduction of the search space.\footnote{
    This observation is based on \cref{fn:realizablebitstringfrac}: out of all $2^{2^{\N}-1}$ bit strings, the realizable ratio falls off doubly-exponentially with $\N$.  On the other hand, for the actual HEI LHS terms, almost half of all bit strings in $\mathcal{D}$ are in fact realizable because of the predominance of pairwise-crossing terms.
}
Nevertheless, this might not be of practical value, since filtering the full collection of bit strings to realizable ones may itself be more computationally costly than searching for a full contraction map for a given HEI candidate (as empirically seems to be the case for the $\N=6$ ones).

\paragraph{Length of LHS bit strings:}
We now turn to the natural possibilities for $\ell$.
While we have specified an entropy inequality with unit coefficients but possibly repeating terms 
\begin{equation}\label{eq:ineq-form2}
    \sum_{n = 1}^{\mL} \ent{X_n} \geq \sum_{n = 1}^{\mR} \ent{Y_n}\,,
\end{equation}
the more traditional way to specify an HEI with as few entropy terms as possible allows for arbitrary positive integer coefficients
\begin{equation}\label{eq:ineq-form-orig}
    \sum_{n = 1}^{\mL'} \alpha_n \, \ent{X_n} \geq \sum_{n = 1}^{\mR'} \beta_n \, \ent{Y_n}\,,
\end{equation}
(with all $X_n$ and $Y_n$ distinct from each other),
where $\alpha_n, \beta_n \in \mathbb{Z}_+$ such that  
\begin{equation}\label{eq:reln_coeffs} 
    \sum_{n = 1}^{\mL'}\alpha_n = \mL \ , \qquad 
    \sum_{n = 1}^{\mR'}\beta_n = \mR \ ,
\end{equation}
so $\mL' \le \mL$ and $\mR' \le \mR$.
In both cases, these forms are specified uniquely, so that there is no ambiguity.\footnote{
    In fact, besides the forms \cref{eq:ineq-form2} and \cref{eq:ineq-form-orig} which make the coefficient convention choice symmetrically between the LHS and RHS, there are two further hybrid cases, wherein we take LHS in the form \cref{eq:ineq-form2} but RHS in the form \cref{eq:ineq-form-orig}, or vice-versa.  In this paragraph we're only focusing on the LHS, but we will return to the RHS distinction below.
}
For the primitive HEIs, while most coefficients are 1, there is a small fraction of terms
 with small integers larger than one, so this is a non-trivial distinction.\footnote{ 
    Out of all non-zero coefficients, taken from the $\N=6$ primitive HEIs constructed in \cite{Hernandez-Cuenca:2023iqh},  about 
    9.2\% of terms come with coefficient 2, 0.26\% with coefficient 3, 0.0016\% with coefficient 4, and none have higher coefficients.
    On the other hand, only about 15\% of HEIs have all coefficients no larger than 1 -- in other words, the remaining HEIs have \emph{some} term with a higher coefficient ($\sim$85\% of HEIs have at least one term with coefficient 2 and $\sim$6\% with coefficient 3).
    }

One natural characterization of contraction maps is then whether we consider the domain having bit strings $x$ with length $\mL$ or $\mL'$, i.e.,
\begin{itemize}[nosep] 
\item (\cmLs): \emph{split} maximally all LHS terms (so that all coefficients $=1$),  $x \in \{0,1\}^{\mL}$
\item (\cmLn): LHS terms remain \emph{non-split} (allowing for coefficients $>1$),   $x \in \{0,1\}^{\mL'}$
\end{itemize} 
In case (\cmLn), we then need to incorporate the coefficients into the Hamming distance as originally defined in \cite{Bao:2015bfa}.

Notice that whenever the cases (\cmLs) and (\cmLn) differ (i.e.\ the LHS has as at least one term with coefficient $>1$), the set of all bit strings $ \mathcal{D}$ necessarily contains some unrealizable ones, namely those whose bits differ on the repeated terms, making cases (\cmf) and (\cmr) genuinely distinct.  On the other hand,
if we restrict to the case (\cmr), so that in (\cmLs) we include only the bit strings  wherein all bits agree on a given subsystem for every subsystem $X_n$, the domains for the two cases (\cmLn) and (\cmLs) are trivial recasting of each other.  For this reason, given that (\cmr) suffices to prove validity of an HEI, we will not explore the distinction between  (\cmLs) and (\cmLn) further here.

\paragraph{Length of RHS bit strings:} 
Analogously to distinguishing split versus non-split cases for the LHS, we can make a similar distinction for the RHS, again using  \cref{eq:ineq-form2} and \cref{eq:ineq-form-orig}:
\begin{itemize}[nosep] 
\item (\cmRs): \emph{split} maximally all RHS terms (so that all coefficients $=1$),  $y \in \{0,1\}^{\mR}$
\item (\cmRn): RHS terms remain \emph{non-split} (allowing for coefficients $>1$),   $y \in \{0,1\}^{\mR'}$
\end{itemize} 
It is easy to see that if a contraction map exists for case (\cmRn), then it necessarily exists for case (\cmRs) simply by copying the bits in the repeated terms, but the converse does not hold (as we will see explicitly in a simple example below, and as was previously noted in  \cite{Bao:2015bfa,Avis:2021xnz}).
A naively natural guess is that all \emph{primitive} HEIs admit the stronger form of contraction, namely case (\cmRn), but while this is works for most primitive HEIs, it is known \cite{Avis:2021xnz} that it does \emph{not} work for all primitive HEIs, already for $\N=5$: in particular, for  \cite[eq.(4.19)]{Avis:2021xnz} (which is a permutation of the 6th inequality in the $\N = 5$ list in \cite{hecdata} and for which (\cmf) and (\cmr) coincide), a contraction map of type (\cmLn,\cmRs) exists whereas one of type  (\cmLn,\cmRn) does not.

Conversely, one might worry that even case (\cmRs) does not suffice to ascertain non-existence of a contraction map.  For example, suppose we consider a candidate inequality in the form \eqref{eq:ineq-form2}, and suppose we can show that no contraction map of type (\cmLs,\cmRs) exists.  One might then still wonder if we multiply both sides by some positive integer $k$, so that instead of $\mL$ terms on LHS we have $k\, \mL$ terms and instead of $\mR$ terms on RHS we have $k\, \mR$ terms, whether now a contraction map might exist.   However, empirically this is not the case.\footnote{
    We thank Sergio Hern\'andez-Cuenca for helpful discussion. 
}
One may also retreat to the hybrid case  (\cmLn,\cmRs).  In fact, \cite[Problem 3]{Avis:2021xnz} speculates whether this would be the formulation guaranteeing completeness.  We revisit this possibility in upcoming work.

\paragraph{Summary of formulations:} 
So we have in principle six distinct possibilities for the setup of contraction map inside
$\{ (\cmf), (\cmr) \} \times \{ (\cmLn), (\cmLs) \}\times \{ (\cmRn), (\cmRs) \}$:
(\cmf,\cmLn,\cmRn),
(\cmf,\cmLn,\cmRs),
(\cmf,\cmLs,\cmRn),
(\cmf,\cmLs,\cmRs),
(\cmr,\cmLn,\cmRn) = (\cmr,\cmLs,\cmRn), and
(\cmr,\cmLn,\cmRs) = (\cmr,\cmLs,\cmRs),
with (\cmf)$\implies$(\cmr) and (\cmRn)$\implies$(\cmRs).
The original contraction map formulated in \cite{Bao:2015bfa} was (\cmf,\cmLn,\cmRn), though it already mentions (\cmf,\cmLn,\cmRs) as a useful variant, while the contraction map formulation we use in our \cref{def:contr_map} is (\cmf,\cmLs,\cmRs).
As shown in \cite{Avis:2021xnz}, (\cmr) is sufficient, and (\cmRn) is not necessary, for validating an HEI.  In practice (\cmf,\cmLn,\cmRs) may actually be the `best bet' in searching for a contraction map.

\subsection{Primitive versus redundant HEIs}
\label{ss:primred}

The previous discussion focused on the context of \emph{primitive} HEIs, but it is also instructive to consider redundant HEIs. 
Primitive HEIs are by definition tighter than redundant ones, so one might naively expect that if a contraction map of a given type exists for two primitive HEIs, it likewise exists their sum.  In this subsection we point out that this is not necessarily the case. Consider an HEI $\mathsf{Q}\ge 0$ which is the sum of two primitive HEIs, $\mathsf{Q}_1\ge 0$ and $\mathsf{Q}_2 \ge 0$.  If we write $\mathsf{Q}=\mathsf{Q}_1 + \mathsf{Q}_2$ in the form \eqref{eq:ineq-form2} and  if we have the (\cmLs,\cmRs) contraction maps for both $\mathsf{Q}_1$ and $\mathsf{Q}_2$, then we can immediately uplift this to a (\cmLs,\cmRs) contraction map for $\mathsf{Q}$ by simply `taking a product' of the two maps.\footnote{
    While this is not the only possible contraction map, it is one which is always guaranteed to work; we leave the proof of the validity as a simple exercise for the reader.}  More specifically, suppose we have contraction maps $f_1$ for $\mathsf{Q}_1$ and $f_2$ for $\mathsf{Q}_2$, such that for respective bit strings $x^{(1)}$ and $x^{(2)}$, we have $f_1(x^{(1)})=y^{(1)}$ and $f_2(x^{(2)})=y^{(2)}$.  Then for $\mathsf{Q}$, for a given bit string $x=(x^{(1)},x^{(2)})$, we define the contraction map $f$ by 
\begin{equation}\label{eq:product_contr}
    f(x)\coloneq (f_1(x^{(1)}),f_2(x^{(2)})) = (y^{(1)},y^{(2)}) \coloneq y \ .
\end{equation}
Similarly, if we write all HEIs in the form \eqref{eq:ineq-form-orig}, but none of the terms $\{X_i,Y_i\}$ in the two HEIs overlap, then the construction of \eqref{eq:product_contr} gives a valid (\cmLn,\cmRn) contraction map $f$ (assuming the (\cmLn,\cmRn) contraction maps $f_1$ and $f_2$ are known).
However, if some of the terms \emph{do} overlap, and according to \eqref{eq:ineq-form-orig} we combine them so as to avoid any repeated terms (so that $\mathsf{Q}$ has fewer terms than the sum of the number of terms of $\mathsf{Q}_1$ and $\mathsf{Q}_2$), then the existence of a (\cmLn,\cmRn) contraction map is no longer guaranteed. 

Let us illustrate this on a simple example of a sum of two MMIs.
Let $\mathsf{Q}_1 = -I_3(1:2:3)$ and $\mathsf{Q}_2 = -I_3(1:2:4)$, so that $\mathsf{Q} = \mathsf{Q}_1 + \mathsf{Q}_2 \ge 0$ is expressed as
\begin{equation}\label{eq:MMIplusMMI}
       2  \ent{12} + \ent{23} + \ent{13}  + \ent{24} + \ent{14} \ge
    { 2 \ent{1} + 2 \ent{2} + \ent{3} + \ent{123}+ \ent{4} + \ent{124} } \ .
\end{equation}
In particular, we see there is one term on LHS and two terms on RHS with coefficient 2, 
where in the form \eqref{eq:ineq-form-orig} we have $\mL'=5$ with 
$X_n =\{  12,23,13,24,14\}$, and  $\mR'=6$ with 
$Y_n =\{ 1, 2,3,123,4,124\}$.  
Now consider the following four bit strings for the LHS, corresponding to the occurrence vectors of $1$ and $2$ in the full $\mathsf{Q}$ and ones originating from those in say $\mathsf{Q}_1$:
\begin{equation}\label{eq:1}
\begin{aligned}
    x_1 &= (1,0,1,0,1) \\
    x_2 &= (1,1,0,1,0) \\
    x_{11} &= (1,0,1,0,0) \\
    x_{12} &= (1,1,0,0,0) 
\end{aligned}
\end{equation}
(These are all allowed bit strings, in other words inside $\mathcal{A}$, so the following argument is relevant in both cases (\cmr) and (\cmf).)
Since $x_1$ and $x_2$ are occurrence vectors, we can immediately write their contraction map elements,
\begin{equation}\label{eq:2}
\begin{aligned}
    y_1 &= (1,0,0,1,0,1) \\
    y_2 &= (0,1,0,1,0,1) 
\end{aligned}
\end{equation}
Let us now try to find a contraction map elements $y_{11}$ and $y_{12}$.
The contraction condition tells us that the weighted Hamming distance has to be no larger for the RHS pair than for the LHS pair, where on the LHS the weights are $\{2,1,1,1,1\}$
whereas on the RHS they are $\{2,2,1,1,1,1\}$.
Since $\abs{x_1 - x_{11}} = 1$, the first two bits of $y_{11}$ must be $(1,0)$ (otherwise $\abs{y_1 - y_{11}} \ge 2$, violating the contraction condition).  
Similarly, since $\abs{x_2 - x_{12}} = 1$, the first two bits of $y_{12}$ must be $(0,1)$.  In other words,
\begin{equation}\label{eq:3}
\begin{aligned}
    y_{11} &= (1,0,\ldots) \\
    y_{12} &= (0,1,\ldots) 
\end{aligned}
\end{equation}
But this makes $\abs{y_{11}-y_{12}}\ge 4$, whereas $\abs{x_{11}-x_{12}}=2$, violating the contraction condition.
This demonstrates that there can be no contraction map in the formulation (\cmRn).
Further insight is gained by asking how does formulation (\cmLs,\cmRs) evade this non-existence, or rather, since its existence is guaranteed, why it cannot be compressed to achieve formulation (\cmLn,\cmRn).  If we split all terms in \eqref{eq:MMIplusMMI} so as to get the form \eqref{eq:ineq-form2} and apply the construction \eqref{eq:product_contr}, we find that in 10 cases where the two LHS $(12)$ columns have the same entry whereas the two $(1)$ columns and/or the two $(2)$ columns on the RHS have a mismatch.  Therefore we cannot compress these terms, so the formulation (\cmRs) cannot directly generate (\cmRn).

Based on this, one might be tempted to conjecture that non-primitive HEIs do not admit a contraction map in formulation (\cmRn).  However we already saw this is not the case if there is no overlap between the terms.  One could then try to weaken the premise by restricting to pairs of HEIs with some overlap.  Indeed, we often find that the overlap on the LHS has fewer terms than overlap on the RHS, just as in our prototypical case of \eqref{eq:MMIplusMMI} where it was 1 and 2, respectively.  However, this is not universally the case: for example for 
$\mathsf{Q}_1 = -I_3(12:3:4)$ and $\mathsf{Q}_2 = -I_3(1:23:5)$, $\mathsf{Q}$ has 1-term overlap ($123$) on LHS and no overlap on RHS. And in fact it \emph{does} admit a contraction map of type (\cmRn).
Given the large spectrum of possibilities of the type of overlap, one needs further restriction. While we leave more detailed exploration of these issues to future work, the basic idea would be to generalize the argument presented for the specific example in \eqref{eq:MMIplusMMI}, ensuring that the relevant bit strings are in $\mathcal{A}$.  
It would also be interesting to restrict the above discussion to the context of centered HEIs where we can use the power of the combinatorial properties discussed in the main text.

\section{Repurification}
\label{sec:repurification}

In this paper, we have been assuming that, in writing a given inequality in the form \eqref{eq:ineq-form}, a party $p\in[\N+1]$, the ``purifier'', has been chosen such that the subsets $X_n,Y_n$ appearing on the left- and right-hand sides do not contain $p$. As discussed in subsec.\ \ref{ssec:HEC}, one can obtain an equivalent inequality by choosing a different party to be the purifier, and this transformation does not affect the property of being an HEI. In this appendix, we will investigate how the other properties of inequalities (or, equivalently, of matrix pairs) are affected by this ``repurification'' transformation.\footnote{Many of the properties discussed in this paper can actually be written quite naturally without choosing a purifier at all. Rather, for each term one can make an independent choice between writing $S(X)$ and writing $S(X^c)$, so that the matrix pair does not necessarily have a row of all-zeros. This leaves a $\mathbb{Z}_2^{m_{\rm L}+m_{\rm R}}$ freedom in how a given inequality or matrix pair is written. It would be interesting to explore possible applications of this larger ``gauge'' freedom.}

Given an inequality with purifier $p$, writing it in terms of a different purifier $i\in[\N+1]\setminus\{p\}$ amounts to the following transformation on the corresponding matrix pair $(\leftm,\rightm)$:
\be\label{newpurifier}
\leftm\mapsto\leftm'= \leftm +_2 e^T
\leftm_{(i)} \,,\qquad
\rightm\mapsto\rightm'=\rightm +_2e^T
\rightm_{(i)}\,,
\ee
where $+_2$ means addition mod 2. In other words, we do a bit-flip on each column in $\leftm$ ($\rightm$) that has a 1 in the row $\leftm_{(i)}$ ($\rightm_{(i)}$). For example, for the pair \eqref{eq:SSA-matrix} corresponding to the strong subadditivity inequality with purifier $p=4$, making party 1 the purifier amounts to the following transformation
\begin{equation}\label{eq:SSA-to-WM}
\leftm =
    \begin{pmatrix}
    1 & 0 \\
    1 & 1 \\
    0 & 1 \\
    0 & 0
    \end{pmatrix},\quad
    \rightm =
    \begin{pmatrix}
    0 & 1 \\
    1 & 1 \\
    0 & 1 \\
    0 & 0
    \end{pmatrix} \qquad \mapsto \qquad \leftm' =  \begin{pmatrix}
    0 & 0 \\
    0 & 1 \\
    1 & 1 \\
    1 & 0
    \end{pmatrix},\quad
    \rightm' =
    \begin{pmatrix}
    0 & 0 \\
    1 & 0 \\
    0 & 0 \\
    0 & 1
    \end{pmatrix}.
\end{equation}
Indeed the resulting inequality $(\leftm',\rightm')$ could have been obtained starting from \eqref{eq:SSA}, then using the property $\ent{X} = \ent{[4] \setminus X}$ on terms $X \ni 1$, to obtain the final transformed inequality:
\begin{equation}\label{eq:WM}
    \ent{34} + \ent{23} \geq \ent{2} + \ent{4}
\end{equation}
(which is an instance of weak monotonicity \eqref{eq:wmo}).

Since bit flips act as isometries on the hypercubes $\{0,1\}^{\mL}$ and $\{0,1\}^{\mR}$, the existence of a contraction map is preserved by this transformation, as expected since the property of being an HEI is preserved. The property of being superbalanced is also preserved (in fact, it may be defined as the property of remaining balanced after any repurification, cf.\ e.g.\ \cite{Hubeny:2018ijt}). On the other hand, for a general matrix pair and general choice of new purifier $i$, the other properties we have discussed in this paper (PCM, BCM, balance, centeredness, and the dominance conditions) are not necessarily preserved, as the reader may verify with simple counterexamples.

However, the situation is different if we focus on centered inequalities, and choose the new purifier accordingly. Given a matrix pair $(\leftm,\rightm)\in\cset$ centered on party $i$, repurifying on $i$ entails bit-flipping the full matrices $\leftm,\rightm$, and results in a pair $(\leftm',\rightm')\in\cset$ that is centered on the previous purifier $p$.\footnote{This transformation should not be confused with switching rows $i$ and $p$ of the matrices (a special case of transformation 2 of subsec.\ \ref{sec:transformations}). While that transformation also results in a matrix pair with purifier $i$ and central region $p$, the effect on the other rows is different: in the latter case they are not affected, while in the former case they are bit-flipped.} For the inequality, this corresponds to writing every term $S(X)$ as $S(X^c)$. This transformation \emph{does} preserve most of the above properties (we don't include balance, which is automatically true for centered inequalities):
\begin{lemma}
The transformation on $\cset$ that involves repurifying on the central party preserves the properties PCM, BCM, dominance, and inclusion dominance.
\end{lemma}
\begin{proof}
If $(\leftm,\rightm)$ obeys PCM then, given a positive vector $v$, $v\leftm\preceq v\rightm$; hence
\be
v\leftm' = |v|e-v\leftm\preceq|v|e-v\rightm=v\rightm'\,.
\ee
Similarly for BCM with a binary vector $v$. Given a dominance map $f$ for $(\leftm,\rightm)$, the following is a dominance map for $(\leftm',\rightm')$:
\be
f'(u):=e-f(e-u)\,;
\ee
in other words, in terms of subsets, if $f$ maps $U$ to $V$, then $f'$ maps $U^c$ to $V^c$. To show that $f'$ is a dominance map, we have, for any bit string $u$ and any party $j$,
\begin{align}
u\cdot\leftm'_{(j)} &= |u|-u\cdot\leftm_{(j)} \\
&=|u|-|\leftm_{(j)}|+(e-u)\cdot\leftm_{(j)} \\
&\le|u|-|\rightm_{(j)}|+f(e-u)\cdot\rightm_{(j)} \\
&=|u|-f'(u)\cdot\rightm_{(j)} \\
& =f'(u)\cdot\rightm'_{(j)}\,.
\end{align}
The same map $f'$ works for inclusion dominance, if $f$ is an inclusion dominance map; alternatively, we can recall that, for centered inequalities, inclusion dominance is equivalent to being a \CPI, which is preserved.
\end{proof}

However, repurifying on the central party does \emph{not} preserve region dominance. A counterexample is provided by the matrix pair of example \ref{regiondom}.

\bibliographystyle{jhep}
\bibliography{references}
\end{document}